\documentclass[a4paper,11pt]{article}
\usepackage[margin=2cm]{geometry}

\usepackage{commands}

\begin{document}

\title{Faster Exponential-Time Approximation Algorithms Using Approximate Monotone Local Search}

\author[1]{Bar\i\c{s} Can Esmer}
\author[1]{Ariel Kulik}
\author[1]{D{\'{a}}niel Marx\thanks{Research supported by the European Research Council (ERC) consolidator grant No.~725978 SYSTEMATICGRAPH.}}
\author[2]{Daniel Neuen}
\author[3]{Roohani Sharma}
\affil[1]{CISPA Helmholtz Center for Information Security, Saarbr\"ucken, Germany. \texttt{\{baris-can.esmer|ariel.kulik|marx\}@cispa.de}}
\affil[2]{Simon Fraser University, Burnaby, Canada. \texttt{dneuen@sfu.ca}}
\affil[3]{Max Planck Institute for Informatics, Saarland Informatics Campus, Saarbr\"ucken, Germany. \texttt{rsharma@mpi-inf.mpg.de}}

\date{}

\begin{titlepage}
\def\thepage{}
\maketitle

\begin{abstract}
We generalize the monotone local search approach of Fomin, Gaspers, Lokshtanov and Saurabh [J.ACM 2019], by establishing a connection between parameterized approximation and exponential-time approximation algorithms for {\em monotone subset minimization} problems.
In a {\em monotone subset minimization} problem the input implicitly describes a non-empty set family over a universe of size $n$ which is closed under taking supersets.
The task is to find a minimum cardinality set in this family.
Broadly speaking, we use {\em approximate monotone local search} to show that a parameterized $\alpha$-approximation algorithm that runs in $c^k \cdot n^{\OO(1)}$ time, where $k$ is the solution size, can be used to derive an $\alpha$-approximation randomized algorithm that runs in $d^n \cdot  n^{\OO(1)}$ time, where $d$ is the unique value in $d\in \left (1, 1+\frac{c-1}{\alpha} \right)$ such that $\D{\frac{1}{\alpha}}{\frac{d-1}{c-1}} =\frac{\ln c }{\alpha}$ and $\D{a}{b}$ is the Kullback-Leibler divergence.
This running time matches that of Fomin et al.\ for $\alpha=1$, and is strictly better when $\alpha >1$, for any $c >1$.
Furthermore, we also show that this result can be derandomized at the expense of a sub-exponential multiplicative factor in the running time.

We use an approximate variant of the exhaustive search as a benchmark for our algorithm.
We show that the classic $2^n \cdot n^{\OO(1)}$ exhaustive search can be adapted to an $\alpha$-approximate exhaustive search that runs in time
$\left ( 1+ \exp\left (-\alpha \cdot \entropy\left (\frac{1}{\alpha}\right)\right)\right)^n \cdot n^{\OO(1)}$, where $\entropy$ is the entropy function.
Furthermore, we provide a lower bound stating that the running time of this $\alpha$-approximate exhaustive search is the best achievable running time in an oracle model.
When compared to approximate exhaustive search, and to other techniques, the running times obtained by approximate monotone local search are strictly better for any $\alpha \geq 1,~c >1$.

We demonstrate the potential of approximate monotone local search by deriving new and faster exponential approximation algorithms for {\sc Vertex Cover}, {\sc $3$-Hitting Set}, {\sc Directed Feedback Vertex Set}, {\sc Directed Subset Feedback Vertex Set}, {\sc Directed Odd Cycle Transversal} and {\sc Undirected Multicut}.
For instance, we get a $1.1$-approximation algorithm for {\sc Vertex Cover} with running time $1.114^n \cdot n^{\OO(1)}$, improving upon the previously best known $1.1$-approximation running in time $1.127^n \cdot n^{\OO(1)}$ by Bourgeois et al.\ [DAM 2011].

\end{abstract}

\end{titlepage}

\section{Introduction}
\label{sec:introduction}
A lot of interesting problems are computationally hard as they do not admit polynomial-time algorithms.
Still, many of them can be solved significantly faster than exhaustive search.
The area of exact exponential algorithms studies the design of such techniques.
Typically, for {\em subset problems}, where the goal is to find a subset of a given $n$-sized universe $U$ that satisfies some property $\Pi$, a solution can be found by enumerating all $2^n$ subsets of $U$.
Therefore, the goal is to design algorithms that beat this exhaustive search and run in time $\OO^*\left ( c^n\right)$\footnote{The $\OO^*$ notation hides polynomial factors in the expression.} for as small $ 1  < c < 2$ as possible.

\paragraph*{Exact Monotone Local Search.}
In a seminal work Fomin, Gaspers, Lokshtanov and Saurabh~\cite{FominGLS19} showed that one can derive faster exact exponential algorithms for subset problems using a parameterized extension algorithm for the problem at hand.
A parameterized extension algorithm for a subset problem additionally takes as input a parameter $k$ and a set $X \subseteq U$, runs in time $\OO^*(f(k))$, and outputs a set $S \subseteq U$ of size at most $k$ such that $S \cup X$ is a solution, if such a set exists.
Fomin et al.~\cite[Theorem~$1.1$]{FominGLS19} showed that if a subset problem
admits a parameterized extension algorithm that runs in time $\OO^*(c^k)$ for some absolute constant $c > 1$, 
then it admits a randomized exact exponential algorithm 
that runs in time $\OO^*\left( \left(\emlsbound(c)\right)^n\right)$, where $\emlsbound(c) = 2 - \frac{1}{c}$.
Their algorithm, called Exact Monotone Local Search (\emls), is simple and is based on monotone local search: it samples a set $X$ of $t$ elements at random, and then extends the set $X$ to an optimum solution using the parameterized extension algorithm.
The non-trivial part of the proof of~\cite[Theorem~$1.1$]{FominGLS19} is to analyze the value of $t$ that optimizes the running time.
This simple algorithm outperforms the exhaustive search for all subset problems that have parameterized extension algorithms running in $\OO^*(c^k)$ time.
Moreover, given the existence of a large number of  problems that admit the desired parameterized extension algorithm, it yields the state-of-the-art exact exponential algorithms for several problems~\cite[Table~$1$]{FominGLS19}. 

\paragraph*{Exponential Approximation.}
Another important algorithmic paradigm which deals with \NP-hardness is the design of {\em approximation} algorithms, which are typically
polynomial-time algorithms that compute a solution which is not necessarily optimum but has a worst-case guarantee on its {\em quality}.
Though several \NP-hard problems admit such algorithms with constant approximation ratios~\cite{Vazirani01}, there are many that 
do not, under reasonable complexity assumptions, for example {\sc Directed Feedback Vertex Set} \cite{Razgon07}.
Also, there are many for which the approximation guarantees cannot be improved beyond a fixed constant.
For example, {\sc Vertex Cover} admits a $2$-approximation but no $(2-\eps)$-approximation under the Unique Games Conjecture~\cite{KhotR08}.

For problems where some hardness of approximation has been established, a natural question is to determine the smallest $c$ such that the barriers of this hardness can be broken by taking $\OO^*(c^n)$ time.
Such algorithms are called {\em exponential approximation} algorithms and this topic has received attention in, e.g., \cite{AroraBS15,BansalCLNN19,BEP11,CyganKW09,MT18}.

Consider {\em subset minimization} problems where the goal is to find a subset of the $n$-sized universe $U$ of {\em minimum cardinality}, 
also called an {\em optimum} solution, that satisfies some additional property $\Pi$.
For any approximation ratio $\alpha \geq 1$, we say that a subset $S \subseteq U$ satisfying the property $\Pi$ is an {\em $\alpha$-approximate solution} if $|S| \leq \alpha \cdot |\OPT|$, where $\OPT \subseteq U$ is an optimum solution. 
An {\em exponential $\alpha$-approximation algorithm} for a subset minimization problem returns an $\alpha$-approximate solution and runs in $\OO^*(c^n)$ time for some $1 < c < 2$.

\paragraph*{Parameterized Approximation.}
A {\em parameterized $\alpha$-approximation} algorithm for a subset minimization problem additionally takes as input the parameter $k$, runs in time $\OO^*(f(k))$, and outputs a solution of size at most $\alpha \cdot k$, if there exists a solution of size at most $k$. 
Analogous to parameterized algorithms, one can define the notion of extension algorithms here (see Section~\ref{sec:results}).
The design of parameterized $\alpha$-approximation algorithms has been an active area of research in the last few years, yielding a plethora of results for problems that 
exhibit some hardness either in the parameterized setting or in the approximation setting
\cite{BrankovicF11,BF13,DvorakFKMTV21,FeigeM06,FeldmannSLM20,FKRS18,GuptaLL18,GuptaLL18b,
 KawarabayashiL20,KS20,Lee19,LokshtanovMRSZ21,Lokshtanov0S20,
 Manurangsi19,Marx08,MarxR09,SkowronF17}.

\paragraph*{Approximate Monotone Local Search (\amls).}
In this paper, we show that one can extend the idea of \emls~\cite{FominGLS19} to \textit{derive faster exponential approximation algorithms from parameterized approximation algorithms}.
Let $\amlsbound(\alpha,c)$ be the unique value in $\left (1, 1+\frac{c-1}{\alpha} \right)$ such that $\D{\frac{1}{\alpha}}{\frac{\amlsbound(\alpha,c)-1}{c-1}} =\frac{\ln c }{\alpha}$ where $\D{a}{b}$ is the Kullback-Leibler divergence defined as  $\D{a}{b}= a \ln \left(\frac{a}{b}\right) +(1-a)\ln \left(\frac{1-a}{1-b}\right)$ (see, e.g.,~\cite{TJ06}).
Our main result can be informally stated as follows.

\begin{mdframed}[backgroundcolor=gray!20] 
 If a monotone subset minimization problem admits a parameterized extension $\alpha$-approximation algorithm that runs in $\OO^*(c^k)$,
 then one can derive a randomized $\alpha$-approximation algorithm that runs in time $\OO^*((\amlsbound(\alpha,c))^n)$ (see Theorem~\ref{thm:main_randomized}).
\end{mdframed}
 
Since $\amlsbound(1,c) = \emlsbound(c) = 2 -\frac{1}{c}$ for every $c >1$, our running time matches that of \emls\ when $\alpha=1$. 

Recall the \emls\ algorithm described earlier. 
The non-trivial part of the proof of~\cite[Theorem~$1.1$]{FominGLS19} is the analysis of the probability that the sampled set is {\em contained} in an optimum solution. 
To obtain our result, we use the same algorithm and show that if we allow the sampled set to also contain {\em some} items from outside the optimum solution,
then one can speed-up the resulting exponential $\alpha$-approximation algorithm. 
Since our analysis need to take into account the calculations for error in the sampled set, this makes analyzing the choice of $t$ more difficult.

In order to better appreciate the running time of our algorithm, that does not seem to have a closed-form formula, we give a mathematical comparison of $\amlsbound(\alpha,c)$ with various benchmark exponential approximation algorithms, showing that our algorithm outperforms all of them.

\paragraph*{Benchmark 1: Brute-Force for Exponential Approximation.}
\label{benchmark-1}
Since exhaustive search is a trivial benchmark against which the running times of (exact) exponential algorithms are measured,
an important question to address is: \textit{how much time does exhaustive search take to find an $\alpha$-approximate solution}?
 
Consider a subset minimization problem that is also monotone, that is, for every $S \subseteq T \subseteq U$, if $S$ is a solution, then $T$ is also a solution.
We show that for monotone subset minimization problems, the classic brute-force approach can be generalized to a (randomized) \emph{$\alpha$-approximation brute-force algorithm} running in time $\OO^*(\brute(\alpha)^n)$,
where $\brute(\alpha) =1+ \frac{(\alpha-1)^{\alpha-1}}{\alpha^\alpha} = 1 + \exp\left(-\alpha \cdot  \HH\left(\frac{1}{\alpha}\right) \right)$ and $\HH(\alpha) = -\alpha \ln \alpha - (1-\alpha) \ln (1-\alpha)$ denotes the entropy function\footnote{We adopt the convention that $0^0=1$.}. 
The running time essentially follows by showing that a uniformly sampled set of $\alpha \cdot |\OPT|$ items is an $\alpha$-approximate solution with probability at least $\brute(\alpha)^{-n}$ (up to polynomial factors, see Theorem~\ref{thm:brute-force-lower}).
The algorithm can also be derandomized at the cost of a sub-exponential factor.
We complement this result by showing that this running time is \emph{best possible}, given only membership oracle access to the problem (Theorem~\ref{thm:brute-force-lower}).
For example, for $\alpha=2$, this brute-force algorithm runs in time $\OO^*(1.25^n)$.
For $\alpha=1.1$, it runs in time $\OO^*(1.716^n)$ (see also Figure \ref{fig:tables}).

\paragraph*{Benchmark 2: Naive conversion from parameterized approximation to exponential approximation.}
\label{benchmark-2}
Yet another upper bound on the running time of exponential approximation algorithms can be given as follows.
Suppose there is a parameterized $\alpha$-approximation for a monotone subset minimization problem that runs in $\OO^*(c^k)$ time.
Run the parameterized algorithm for every value of the parameter $k$ between $0$ to $\frac{n}{\alpha}$, and return the solution of minimum cardinality among the solutions returned by the parameterized algorithm.
If no solution was found by the parameterized algorithm, then return the whole universe.
It can be easily verified that this indeed yields an $\alpha$-approximation algorithm with running time $\OO^*(\left(\naive(\alpha,c)\right)^n)$, where $\naive(\alpha,c)=c^{\frac{1}{\alpha}}$.
Observe that for large values of $c$ and appropriate $\alpha$, $\naive(\alpha,c)$ could be much larger than even $2$. 
But for smaller values of $c$, it could sometimes beat the brute-force approximation (see Section~\ref{sec:applications}).

\paragraph*{Benchmark 3: \emls\ in the approximate setting.}
\label{benchmark-3}
From the description of the \emls\ algorithm, it is not difficult to deduce that given a parameterized extension $\alpha$-approximation for a monotone subset minimization problem that runs in time $\OO^*(c^k)$, one can derive an exponential $\alpha$-approximation for the problem that runs in time $\OO^*((\emlsbound(c))^n)$. 
This trivial generalization of \emls\ to the approximate setting already performs better than the naive conversion in cases when $c$ is small. 
For example, {\sc Vertex Cover} has a parameterized (extension) $1.1$-approximation algorithm that runs in time $\OO^*(1.1652^k)$~\cite{BEP11}.
\emls\ gives a $1.1$-approximation running in time $\OO^*(1.1417^n)$ whereas the naive conversion gives a running time of $\OO^*(1.1462^n)$.

\paragraph*{Comparisons.}
As stated earlier, we show that (see Lemma~\ref{lem:comparison}) \amls\ is {\em strictly} faster than the brute-force algorithm (\hyperref[benchmark-1]{Benchmark 1}) or the naive-conversion approach (\hyperref[benchmark-2]{Benchmark 2}) described earlier, for every $\alpha \geq 1$ and $c > 1$.
In fact, Lemma~\ref{lem:comparison} shows that $\amlsbound(\alpha,c)$ converges to $\brute(\alpha)$ as $c \to \infty$, which would be the expected behavior, because as the parameterized algorithm becomes ``less useful'', the running time of our algorithm gets closer to the running time possible without the use of any problem-specific algorithm.
Since $\amlsbound(1,c) = \emlsbound(c) =2-\frac{1}{c}$, Lemma~\ref{lem:comparison} also shows that the running time is strictly better than that of \emls\ (\hyperref[benchmark-3]{Benchmark 3}) when $\alpha >1$.

\begin{lemma}[$\star$\footnote{The proofs of statements marked with $\star$ appear in the Appendix.}]
	\label{lem:comparison}
	For every $c > 1$ the following holds:
	\begin{enumerate}
	\item\label{item:comparison-1} $\amlsbound(\alpha,c)< \min\{\brute(\alpha), \naive(\alpha,c)\}$ for every $\alpha \geq 1$. In fact, $\amlsbound(\alpha,c) \xrightarrow[c \to \infty]{} \brute(\alpha)$.
	\item\label{item:comparison-2} $\amlsbound(\alpha,c) < \emlsbound(c)$ for every $\alpha >1$. In fact, $\amlsbound(\alpha,c)$ is a strictly decreasing function of $\alpha$. 
	\end{enumerate}
\end{lemma}

\paragraph*{Applications and Derandomization.}
We show in Section~\ref{sec:applications} that \amls\ can be used to derive new and faster exponential approximation algorithms for {\sc Vertex Cover}, {\sc $3$-Hitting Set}, {\sc Directed Feedback Vertex Set}, {\sc Directed Subset Feedback Vertex Set}, {\sc Directed Odd Cycle Transversal} and {\sc Undirected Multicut}.
We also show in Section~\ref{sec:derand-short} that, as in~\cite{FominGLS19}, our algorithm can be derandomized at the expense of a multiplicative sub-exponential factor in the running time.

\section {Definitions and Our Results}
\label{sec:results}
\subsection{Formal Definitions} 

We now give some formal definitions that will be required to formally state and describe our main results.
An {\em implicit set system} is a function $\Phi$ that takes as input a string $I \in \{0,1\}^{\star}$, called an {\em instance}, and returns a set system $(U_I, \F_I )$ where $U_I$ is a universe and $\F_I $ is a collection of subsets of $U_I$.  
We use $n$ to denote the size of the universe, that is, $n \coloneqq |U_I|$.
We say that an implicit set system $\Phi$ is {\em polynomial-time computable} if there are two polynomial-time algorithms, the first one, given $I \in \{0,1\}^{\star}$, computes the set $U_I$,
and the second one, given $S \subseteq U_I$, correctly decides if $S \in \F_I$.  
A family $\F \subseteq 2^U$ is called {\em monotone} if $U \in \F$ and for every $S \subseteq T \subseteq U$, if $S \in \F$, then $T \in \F$.
We say that an implicit set system $\Phi$ is {\em monotone} if $\F_I$ is monotone for every input $I$.
Throughout the remainder of this work, we only deal with implicit set systems that are polynomial-time computable and monotone.
So for the sake of convenience, we refer to a polynomial-time computable and monotone implicit set system simply as an implicit set system.

For an implicit set system $\Phi$, the problem $\phiminsubset$ takes as input a string $I \in \{0,1\}^{\star}$ and asks to find $S \in \F_I$ such that $|S|$ is minimum.
We refer to the sets in $\F_I$ as solutions of $I$ and we call the sets in $\F_I$ of minimum cardinality as minimum solutions or optimal solutions of $I$.
For $\alpha \geq 1$, we say that an algorithm is a (randomized) $\alpha$-approximation algorithm for $\phiminsubset$ if on input $I$, it returns a set $S \in \F_I$ such that 
$|S| \leq \alpha \cdot |\OPT|$ (with a constant probability), where $\OPT$ is an optimum solution of $\phiminsubset$. 

One can observe that many fundamental graph problems, such as {\sc Vertex Cover}, {\sc Feedback Vertex Set},  {\sc Directed Feedback Vertex Set}, etc., can be cast as a $\phiminsubset$ problem.
Consider for example the {\sc Vertex Cover} problem.
Given a graph $G=(V,E)$ we say a subset $S \subseteq V$ is a {\em vertex cover} if for every $(u,v) \in E$ it holds that $u \in S$ or $v\in S$. 
The input for the {\sc Vertex Cover} problem is a graph $G$ and the objective is to find a vertex cover $S$ of $G$ such that $|S|$ is minimum.
We can cast {\sc Vertex Cover} as a $\phiminsubset$ problem for the implicit set system $\phiVC$ defined as follows.
The instance of the problem is interpreted as a graph $G=(V,E)$.
We define the universe $U_G$ as the set of vertices $V$ and the set of solutions $\F_G=\{S \subseteq V \mid S \textnormal{ is a vertex cover of } G\}$ is the set of all vertex covers of $G$.
Finally, we define $\phiVC(G)=(U_G,\F_G)$.
It can be easily verified that $\phiVC$ is an implicit set system (i.e., it is polynomial-time computable and monotone). 

We say an algorithm is a {\em parameterized (randomized) $\alpha$-approximate $\Phi$-extension} if, given an instance $I \in \{0,1\}^{\star}$, $X \subseteq U_I$ and a parameter $k \in \mathbb{N}$, it returns a set $Y \subseteq U_I$ which satisfies the following property (with a constant probability):
 if there exists a set $S \subseteq U_I$ such that $S \cup X \in \F_I$ and $\abs{S} \leq k$, then it holds that $Y \cup X \in \F_I$ and $\abs{Y} \leq \alpha \cdot k$. 
We use the shorthand {\em (randomized) $(\alpha,\Phi)$-extension algorithm} to refer to a {\em parameterized (randomized) $\alpha$-approximate $\Phi$-extension algorithm}.
Observe that a parameterized $\alpha$-approximation algorithm for {\sc Vertex Cover} can be turned into  an $(\alpha,\phiVC)$-extension algorithm with the same running time, by taking the instance $(G,X,k)$ of the $(\alpha,\phiVC)$-extension algorithm, and running the parameterized $\alpha$-approximation algorithm on the instance $(G-X,k)$.
Observe that this way of converting parameterized $\alpha$-approximation algorithms to $(\alpha,\Phi)$-extension algorithms holds for various implicit set systems $\Phi$, for example, when $\Phi$ corresponds to a vertex deletion problem to a hereditary graph class.

\subsection{Our results}
Given an $(\alpha,\Phi)$-extension algorithm with running time $\OO^*(c^k)$ we design an $\alpha$-approximation algorithm for $\phiminsubset$ with running time $\OO^*(\amlsbound(\alpha,c)^n)$ where $\amlsbound$ is defined 
as the unique value $\gamma \in \left(1, 1+\frac{c-1}{\alpha}\right)$ such that $\D{\frac{1}{\alpha}}{\frac{\gamma-1}{c-1}} =\frac{\ln c }{\alpha}$.
Note that $\amlsbound(\alpha,c)$ is indeed well-defined because for every $\alpha \geq 1$, the function $f(\delta) \coloneqq \D{\frac{1}{\alpha}}{\delta}$ is monotonically decreasing in the interval $\delta \in \left(0, \frac{1}{\alpha}\right)$ as well as $f(\delta) \xrightarrow[\delta \to 0]{} \infty$ and $f(\delta) \xrightarrow[\delta \to \frac{1}{\alpha}]{} 0$. 

\begin{theorem}[Approximate Monotone Local Search]
 \label{thm:main_randomized}
 Let $\Phi$ be an implicit set system and $\alpha \geq 1$.
 If there is a randomized $(\alpha,\Phi)$-extension algorithm that runs in time $\OO^*(c^k)$,
 then there is a randomized $\alpha$-approximation algorithm for $\phiminsubset$ that runs in time $\OO^*((\amlsbound(\alpha,c))^{n})$.
\end{theorem}

The formula for $\amlsbound(\alpha,c)$ (which describes the running time of Theorem~\ref{thm:main_randomized}) is not a closed-form formula, and we do not expect a closed-form formula for general $\alpha,c$ to exist.
However, it represents a {\em tight} analysis of our algorithm. 
Despite being represented as an implicit formula, its basic properties can be deduced (see Lemma~\ref{lem:comparison}).
Also, $\amlsbound$ can be easily evaluated for every $\alpha, c > 1$.
Indeed, for every $\alpha \geq 1$, the function $f(\gamma)= \D{\frac{1}{\alpha}} {\frac{\gamma-1}{c-1}}$ is monotonically decreasing in the interval $\left(1, 1+\frac{c-1}{\alpha}\right)$.
This means that $\amlsbound(\alpha,c)$ can be evaluated to an arbitrary precision, for every $\alpha \geq 1$ and $c>1$, using binary search.
In particular, the running time implied by Theorem~\ref{thm:main_randomized} can be evaluated.

Theorem~\ref{thm:main_randomized} can be used to obtain faster (than the state-of-art) exponential approximation algorithms for some $\phiminsubset$ problems. 
For example, the brute-force $1.1$-approximation algorithm runs in time $\OO^*(\brute(1.1)^n) = \OO^*(1.716^n)$.
The best parameterized $1.1$-approximation for {\sc Vertex Cover} runs in time $\OO^*(1.1652^k)$~\cite{KS20}, where $k$ is the parameter.
Using the naive conversion, this algorithm can be naively converted to a $1.1$-approximation that runs in time $\OO^*(\naive(1.1, 1.1652)^n) = \OO^*(1.149^n)$.
The previously known fastest $1.1$-approximation algorithm for {\sc Vertex Cover} runs in time $\OO^*(1.127^n)$ ~\cite{BEP11}. 
Using Theorem~\ref{thm:main_randomized} in conjunction with the $\OO^*(1.1652^k)$ algorithm of~\cite{KS20}, 
we get a $1.1$-approximation algorithm for {\sc Vertex Cover} with running time $\OO^*(1.114^n)$, improving over all of the above.
We provide additional applications in Section~\ref{sec:applications}.

In Section~\ref{sec:derand-short}, we show that the algorithm of Theorem~\ref{thm:main_randomized} can be derandomized, at the cost of a sub-exponential factor in the running time, by generalizing the construction of set inclusion families from~\cite{FominGLS19}.
 
\begin{theorem}[Derandomization Approximate Monotone Local Search]
 \label{thm:main_derandomized}
 Let $\Phi$ be an implicit set system and $\alpha \geq 1$.
 If there is an $(\alpha,\Phi)$-extension algorithm that runs in time $\OO^*(c^k)$, then there is an $\alpha$-approximation algorithm for $\phiminsubset$ that runs in time $\OO^*\left(\left(\amlsbound(\alpha,c)\right)^{n +o(n)}\right)$.
\end{theorem}

\subsection{Approximate Monotone Local Search}
We now present the algorithm underlying Theorem~\ref{thm:main_randomized} and give a sketch for its analysis.
Recall $\Phi$ is the implicit set family and $\alpha \geq 1$ is the approximation ratio.
Let $\Aext$ denote the $(\alpha,\Phi)$-extension algorithm that runs in time $\OO^*(c^k)$ where $k$ is the parameter.
Given an instance $I \in \{0,1\}^{\star}$, let $\Phi(I)=(U_I,\F_I)$.
The algorithm for Theorem~\ref{thm:main_randomized} is described in Algorithm~\ref{algo:final}, and it is denoted by \amlsalgo.
It uses the subroutine \sample\ (Algorithm~\ref{algo:intermediary}) which samples a random set from $U_I$, which is subsequently extended, using $\Aext$, to yield the solution. 
Algorithm~\ref{algo:final} coincides with the algorithm of~\cite[Theorem~$1.1$]{FominGLS19} when $\alpha=1$.
 
\begin{algorithm}
	\begin{algorithmic}[1]
		\Input $I \in \{0,1\}^{\star} , k \in \mathbb{N}, t \in \mathbb{N}$
		\State Sample a set $X$ of size $t$ from $U_I$ uniformly at random. \label{int:select}
		\State  $Y \gets \Aext\left( I, X, k - \left\lceil\frac{t}{\alpha}\right\rceil \right) $. \label{int:Aext}
		\State $Z \gets X \cup Y$. \label{int:T}
		\State If $Z \in \F_I$ and $\abs{Z} \leq \alpha \cdot k$, then {\bf return} $Z$, otherwise {\bf return} $U_I$.
	\end{algorithmic}
	\caption{$\sample(I,k,t)$}\label{algo:intermediary}
\end{algorithm}

Let $\OPT$ be an optimum solution of the instance $I$ of $\phiminsubset$.
Consider the execution of \sample\ on the instance $(I,k,t)$ where $k=|\OPT|$.
In Step~\ref{int:select} of \sample\ if $|X\cap \OPT| \geq \frac{t}{\alpha}$ then $|\OPT \setminus X| \leq k- \frac{t}{\alpha}$. 
Therefore, in Step~\ref{int:Aext} $\Aext$ must return a set $Y$ such that $X\cup Y\in \F_I$ and $|Y| \leq \alpha \cdot (k-\frac{t}{\alpha}) = \alpha k -t$. 
Thus, the set $Z = X \cup Y$ computed in Step~\ref{int:T} is an $\alpha$-approximate solution of $I$.
Let $\p(n,k,t,x)$ be the probability that a uniformly random set $X$ of $t$ items out of $[n] \coloneqq \{1,\ldots,n\}$ satisfies $|X \cap [k]| \geq x$.
The distribution of $|X \cap [k]|$ is commonly referred as {\em hyper-geometric}.
Since $\Pr\left(|X\cap \OPT| \geq \frac{t}{\alpha}\right)=\p(n,k,t,\frac{t}{\alpha})$, \sample\ returns an $\alpha$-approximate solution of $I$ with probability $\p(n,k,t,\frac{t}{\alpha})$.
Observe that the running time of the \sample\ subroutine is proportional (up to polynomial factors) to the running time of the call to $\Aext$ in Step~\ref{int:Aext}.
Thus, the \sample\ subroutine runs in time $\OO^*(c^{k - \frac{t}{\alpha}})$. 

\begin{algorithm}
	\begin{algorithmic}[1]
		\Input $I \in \{0,1\}^{\star}$
		\State Define $n = |U_I|$ and $\sol\leftarrow\emptyset$.
		\For{\label{amls:loop}$k$ from $0$ to $\frac{n}{\alpha}$}
		\State $\displaystyle t\leftarrow \argmin_{ t \in \left[ 0,\alpha k\right]\cap \mathbb{N} } \left(\frac{c^{k-\frac{t}{\alpha}}}{\p\left(n,k,t,\frac{t}{\alpha}\right)}\right)$. \label{amls:select_t}
		\State Run $\sol = \sol \cup \left\{ \sample(I,k,t)\right\}$ for $\left(\p\left( n, k, t, \frac{t}{\alpha}\right)\right)^{-1}$ times. \label{amls:call_sample}
		\EndFor
		\State {\bf Return} a minimum sized set in $\sol$.
	\end{algorithmic}
\caption{$\amlsalgo(I)$ }\label{algo:final}
\end{algorithm}

Consider the execution of \amlsalgo\ with the input $I$. The analysis of \amlsalgo\
focuses on the iteration of Step~\ref{amls:loop} when $k = |\OPT|$.
In this iteration, each call to $\sample(I,k,t)$ returns an $\alpha$-approximate solution with probability $\p(n,k,t,\frac{t}{\alpha})$ (as argued above). 
Since in Step~\ref{amls:call_sample}, $\sample$ is invoked $\left(\p(n,k,t,\frac{t}{\alpha})\right)^{-1}$ times, at the end of the execution of Step~\ref{amls:call_sample}, the set $\sol$ contains an $\alpha$-approximate solution of $I$ with a constant probability.

The running time of a {\em fixed iteration} in Step~\ref{amls:loop} of \amlsalgo\ is $(\p(n,k,t,\frac{t}{\alpha}))^{-1}$ times the running time of \sample, that is, $\frac{c^{k-\frac{t}{\alpha}}}{\p(n,k,t,\frac{t}{\alpha})}$. 
Let us denote $\iter_{n,k,c}(t) =\frac{c^{k-\frac{t}{\alpha}}}{\p(n,k,t,\frac{t}{\alpha})} $.
Observe that the value of $t$ selected in Step~\ref{amls:select_t} minimizes $\iter_{n,k,c}(t)$.
From the algorithmic perspective the selection of the optimal value of $t$ is straightforward as
$\iter_{n,k,c}(t)$ can be computed in polynomial time for each value of $t$ (given $n$ and $k$).
However, the asymptotic analysis of $\iter_{n,k,c}(t)$, and hence the overall running time, requires an in-depth  understanding of the random process and serves as the main technical contribution of this paper.

As described earlier, when $\alpha=1$, Algorithm~\ref{algo:final} coincides with the algorithm of~\cite[Theorem~$1.1$]{FominGLS19}. 
The analysis of~\cite[Theorem~$1.1$]{FominGLS19} lower bounds the probability that the set $X$ sampled on Step~\ref{int:select} of \sample\ satisfies $X \subseteq \OPT$.
The analysis of our algorithm lower bounds the probability that $\frac{|X\cap \OPT|}{|X|}\geq \frac{1}{\alpha}$.
In particular, the sampling step may select items which are not in $\OPT$, though the number of such items is restricted.
This allows for an improved running time in comparison to that of~\cite{FominGLS19} (see Lemma~\ref{lem:comparison}), but renders the analysis of the running time to be more involved.

For the analytical estimation of $t$, which is selected in Step~\ref{amls:select_t} of \amlsalgo, 
the question that one needs to understand is that {\em how many items should the algorithm sample before it decides to use $\Aext$ to extend the sampled set}.
Assume that the algorithm already sampled a set $X$ of $t$ items such that $|X\cap \OPT|\approx \frac{t}{\alpha}$.
Let $\varepsilon>0$ be some small number.
Observe that $U_I \setminus X$ contains $\approx k-\frac{t}{\alpha}$ items from $\OPT$, and thus $\frac{ \left| (U_I \setminus X) \cap \OPT \right| }{| U_I \setminus X| }\approx \frac{k-\frac{t}{\alpha}}{n-t}$. 
The algorithm now has two options: it can either further sample a set $A$ of additional $\varepsilon \cdot n$ items or, use $\Aext$ with the parameter $k-\frac{t}{\alpha}$ to extend $X$ to a final solution. 
In the first case, the time taken to sample a set $A$ of $\varepsilon \cdot n$ items such that $|A \cap \OPT| \geq \frac{|A|}{\alpha} = \frac{\varepsilon \cdot n}{\alpha}$ holds with constant probability, is $\left(\Pr\left( |A\cap \OPT| \geq \frac{\varepsilon \cdot n}{\alpha} \right)\right)^{-1}$.
In the second case, the algorithm spends an additional factor of $c^{\frac{\varepsilon \cdot n}{\alpha}}$ time to extend the set $X$, instead of $X \cup A$, to the final solution.
Thus, if $\Pr\left( |A\cap \OPT| \geq \frac{\varepsilon \cdot n}{\alpha} \right)>c^{-\frac{\varepsilon\cdot n}{\alpha}}$, it is better to continue sampling, and otherwise, it is better to run $\Aext$ on the instance $(I,X,k-\frac{t}{\alpha})$.
Therefore, to understand the analytics of the chosen $t$, one needs to upper bound $\Pr\left( |A\cap \OPT| \geq \frac{\varepsilon \cdot n}{\alpha} \right)$.
 
We view the sampling of $A$ as an iterative process in which the items are sampled one after the other.
When sampling the $\ell$-th item, the ratio between the remaining items in $\OPT$ and the available items is $\approx \frac{k-\frac{t}{\alpha}-\Delta}{n-t-\ell}$, where $\Delta$ is the number of items from $\OPT$ sampled in previous iterations and $0\leq \Delta \leq \ell\leq \varepsilon\cdot  n$. 
As $\varepsilon \cdot n$ is small, we estimate $\frac{k-\frac{t}{\alpha}-\Delta}{n-t-\ell} \approx \frac{k-\frac{t}{\alpha}}{n-t}$. 
Therefore, the probability that the $\ell$-th sampled item is in $\OPT$ is roughly $\frac{k-\frac{t}{\alpha}}{n-t}$.
Thus, $|A\cap \OPT|$ can be estimated as the sum of $\varepsilon \cdot n$ Bernoulli random variables $\textsf{x}_1,\ldots, \textsf{x}_{\varepsilon\cdot n}$  with probability $\frac{k-\frac{t}{\alpha}}{n-t}$.
Thus,
$$\Pr\left( |A\cap \OPT| \geq \frac{\varepsilon \cdot n}{\alpha} \right) \approx \Pr\left( \sum_{i=1}^{\varepsilon \cdot n} \textsf{x}_i \geq \frac{\varepsilon \cdot n}{\alpha}\right)  \approx \exp \left( - \varepsilon \cdot n \cdot \D{\frac{1}{\alpha}}{ \frac{k-\frac{t}{\alpha}}{n-t}} \right),$$
where the last estimation follows from a large deviation property of binomial distributions~\cite[Theorem~$11.1.4$]{TJ06} and assumes $|\OPT|\leq\frac{n}{\alpha}$. 

Therefore, the (optimal) selection of $t$ which minimizes $\iter_{n,k,c}(t)$ is the largest $t$ which satisfies $\exp \left( - \varepsilon \cdot n \cdot  \D{\frac{1}{\alpha}}{ \frac{k-\frac{t}{\alpha}}{n-t}} \right) > c^{-\frac{\varepsilon \cdot n}{\alpha}} $, or equivalently, $ \D{\frac{1}{\alpha}}{ \frac{k-\frac{t}{\alpha}}{n-t}} \approx \frac{ \ln c }{\alpha}$.
We use this value of $t$ to bound the running time of an iteration of Step~\ref{amls:loop}, that is, to upper bound $\displaystyle{\min_{t \in [0,\alpha k] \cap \mathbb{N}} \iter_{n,k,c}(t)}$. 
This analytical estimation of $t$ forms the crux in analyzing the overall running time of Algorithm~\ref{algo:final}.

\subsection{Applications of \amls}\label{sec:applications}

\begin{figure}
	\centering
	\begin{subfigure}{.5\textwidth}
		\centering
		\caption{{\sc Vertex Cover}}
		\begin{tikzpicture}[scale = 1.0]
	\begin{axis}[
		xmin = 1, xmax=2, ymin =0.99 , ymax=1.3, xlabel= 	{approximation ratio}, 
		ylabel={exponent base}, samples=50]
		
		\addplot[mark=*, mark options={scale=0.3}, only marks, black, ] coordinates {
			(1.1, 1.127)
			(1.2, 1.099)
			(1.3, 1.083)
			(1.4,1.069)
			(1.5, 1.056)
			(1.6, 1.043)
			(1.7, 1.032)
			(1.8, 1.021)
			(1.9, 1.010)
		};
	\addplot[brown, thick] coordinates {
		( 1.000000001 , 1.2729999991859242 )
		( 1.01 , 1.2649653094298616 )
		( 1.02 , 1.2571373939712263 )
		( 1.03 , 1.2442282001309801 )
		( 1.04 , 1.2261445864110327 )
		( 1.05 , 1.2100681794588304 )
		( 1.06 , 1.1955890378686849 )
		( 1.07 , 1.1824295044570683 )
		( 1.08 , 1.1703882123659333 )
		( 1.09 , 1.1593121081105104 )
		( 1.1 , 1.1490808122780696 )
		( 1.11 , 1.1395971476039108 )
		( 1.12 , 1.1307810452809253 )
		( 1.13 , 1.1225654347204765 )
		( 1.1400000000000001 , 1.1148933647520296 )
		( 1.15 , 1.107715922760024 )
		( 1.16 , 1.100990691170586 )
		( 1.17 , 1.0946805770850985 )
		( 1.18 , 1.0888494487604061 )
		( 1.19 , 1.0838275019055315 )
		( 1.2 , 1.0791028639600801 )
		( 1.21 , 1.0746523661516458 )
		( 1.22 , 1.0704553869156828 )
		( 1.23 , 1.0664934845547733 )
		( 1.24 , 1.062750095403748 )
		( 1.25 , 1.0592102837133561 )
		( 1.26 , 1.0558605324432717 )
		( 1.27 , 1.0526885675929298 )
		( 1.28 , 1.0496832090290649 )
		( 1.29 , 1.0468342437111902 )
		( 1.3 , 1.044132316913298 )
		( 1.31 , 1.041568838529108 )
		( 1.32 , 1.0391359022250646 )
		( 1.33 , 1.0368262148112377 )
		( 1.34 , 1.0346330347980015 )
		( 1.35 , 1.0325501185949328 )
		( 1.3599999999999999 , 1.0305716730979289 )
		( 1.37 , 1.0286923137190385 )
		( 1.38 , 1.0269070275768706 )
		( 1.3900000000000001 , 1.0252111403084267 )
		( 1.4 , 1.0236002868997656 )
		( 1.4100000000000001 , 1.0220703854688933 )
		( 1.42 , 1.0206176138477516 )
		( 1.43 , 1.0192383882314038 )
		( 1.44 , 1.0179293445957283 )
		( 1.45 , 1.0166873211663572 )
		( 1.46 , 1.0155093432564755 )
		( 1.47 , 1.0143099005514151 )
		( 1.48 , 1.013115850600237 )
		( 1.49 , 1.0119908647439673 )
		( 1.5 , 1.011014695897737 )
		( 1.51 , 1.0102145898097499 )
		( 1.52 , 1.0094561998661653 )
		( 1.53 , 1.0087378291424574 )
		( 1.54 , 1.0080578693065012 )
		( 1.55 , 1.0073752256533348 )
		( 1.56 , 1.0067066312521522 )
		( 1.57 , 1.006089181772777 )
		( 1.58 , 1.0055964172648595 )
		( 1.5899999999999999 , 1.0051309156410753 )
		( 1.6 , 1.0046916469020624 )
		( 1.6099999999999999 , 1.0042546814982978 )
		( 1.62 , 1.0038330201404557 )
		( 1.63 , 1.0034750995846253 )
		( 1.6400000000000001 , 1.0031566602259125 )
		( 1.65 , 1.002856023661818 )
		( 1.6600000000000001 , 1.0025550764443678 )
		( 1.67 , 1.0022896483701074 )
		( 1.6800000000000002 , 1.0020596321603543 )
		( 1.69 , 1.001838839224823 )
		( 1.7000000000000002 , 1.001627635671499 )
		( 1.71 , 1.001451325037686 )
		( 1.72 , 1.001282975215024 )
		( 1.73 , 1.0011274331316191 )
		( 1.74 , 1.0009932626689293 )
		( 1.75 , 1.0008634266452612 )
		( 1.76 , 1.0007541066094836 )
		( 1.77 , 1.0006500847622903 )
		( 1.78 , 1.0005608401437944 )
		( 1.79 , 1.0004791950765761 )
		( 1.8 , 1.000406280717164 )
		( 1.81 , 1.000343047144937 )
		( 1.82 , 1.0002870458955442 )
		( 1.83 , 1.0002378462215087 )
		( 1.8399999999999999 , 1.0001949698794221 )
		( 1.85 , 1.000157895424352 )
		( 1.8599999999999999 , 1.0001264697055459 )
		( 1.87 , 1.0000995486685114 )
		( 1.88 , 1.0000770400874741 )
		( 1.8900000000000001 , 1.0000584343759045 )
		( 1.9 , 1.0000431924372712 )
		( 1.9100000000000001 , 1.000031008217195 )
		( 1.92 , 1.0000214405870749 )
		( 1.9300000000000002 , 1.0000141373156854 )
		( 1.94 , 1.0000087677812217 )
		( 1.9500000000000002 , 1.0000049958854522 )
		( 1.96 , 1.0000025187882209 )
		( 1.97 , 1.000001046455369 )
		( 1.98 , 1.0000002563600316 )
		( 1.99 , 1.0000000000119482 )
		( 1.9999999 , 1.0000000000000027 )
	};
	\addplot[red, thick] coordinates {
		( 1.000000001 , 1.214454040924442 )
		( 1.01 , 1.2026317346611575 )
		( 1.02 , 1.1937657772241668 )
		( 1.03 , 1.1829188563622037 )
		( 1.04 , 1.1696458059475097 )
		( 1.05 , 1.1579462570578425 )
		( 1.06 , 1.147468219436361 )
		( 1.07 , 1.1379791088464872 )
		( 1.08 , 1.1293140538370485 )
		( 1.09 , 1.1213503458722733 )
		( 1.1 , 1.1139933192777192 )
		( 1.11 , 1.107167900194372 )
		( 1.12 , 1.1008132346690174 )
		( 1.13 , 1.0948791092018053 )
		( 1.1400000000000001 , 1.0893234738801483 )
		( 1.15 , 1.0841106736337345 )
		( 1.16 , 1.0792101518371808 )
		( 1.17 , 1.0745954786564162 )
		( 1.18 , 1.0703091342871418 )
		( 1.19 , 1.0665783168622305 )
		( 1.2 , 1.0630574655658993 )
		( 1.21 , 1.0597301405541362 )
		( 1.22 , 1.0565817664812673 )
		( 1.23 , 1.0535993568055084 )
		( 1.24 , 1.0507712883745994 )
		( 1.25 , 1.0480871149997306 )
		( 1.26 , 1.0455374126099055 )
		( 1.27 , 1.0431136492121342 )
		( 1.28 , 1.0408080752957507 )
		( 1.29 , 1.038613630734035 )
		( 1.3 , 1.0365238652554023 )
		( 1.31 , 1.0345328702799048 )
		( 1.32 , 1.0326352201762867 )
		( 1.33 , 1.0308259210402704 )
		( 1.34 , 1.0291003665116905 )
		( 1.35 , 1.0274542989618507 )
		( 1.3599999999999999 , 1.0258837755476122 )
		( 1.37 , 1.0243851381010516 )
		( 1.38 , 1.0229549869935357 )
		( 1.3900000000000001 , 1.0215901574250275 )
		( 1.4 , 1.0202876987862437 )
		( 1.4100000000000001 , 1.019044856086777 )
		( 1.42 , 1.0178590534151981 )
		( 1.43 , 1.016727878933355 )
		( 1.44 , 1.0156490717946534 )
		( 1.45 , 1.0146205098657965 )
		( 1.46 , 1.0136401991245751 )
		( 1.47 , 1.0126380694307486 )
		( 1.48 , 1.0116354667694532 )
		( 1.49 , 1.0106851242007866 )
		( 1.5 , 1.009854855109447 )
		( 1.51 , 1.009169673025097 )
		( 1.52 , 1.008517356046676 )
		( 1.53 , 1.0078967236377934 )
		( 1.54 , 1.0073066579010517 )
		( 1.55 , 1.0067118670363557 )
		( 1.56 , 1.0061265838885447 )
		( 1.57 , 1.0055832057318472 )
		( 1.58 , 1.0051469606907701 )
		( 1.5899999999999999 , 1.004733166059971 )
		( 1.6 , 1.0043410848000884 )
		( 1.6099999999999999 , 1.0039495335877802 )
		( 1.62 , 1.0035700380927524 )
		( 1.63 , 1.0032463166259005 )
		( 1.6400000000000001 , 1.002957072958094 )
		( 1.65 , 1.0026829540412958 )
		( 1.6600000000000001 , 1.0024075225174534 )
		( 1.67 , 1.0021635452392368 )
		( 1.6800000000000002 , 1.0019512536804442 )
		( 1.69 , 1.00174674451058 )
		( 1.7000000000000002 , 1.00155038693793 )
		( 1.71 , 1.001385817279926 )
		( 1.72 , 1.0012281425564316 )
		( 1.73 , 1.0010819462416165 )
		( 1.74 , 1.0009553883169675 )
		( 1.75 , 1.0008325188992617 )
		( 1.76 , 1.0007287076093345 )
		( 1.77 , 1.0006296184137997 )
		( 1.78 , 1.000544330984212 )
		( 1.79 , 1.0004660703474075 )
		( 1.8 , 1.000395967642424 )
		( 1.81 , 1.0003349938369044 )
		( 1.82 , 1.0002808423901166 )
		( 1.83 , 1.0002331376113212 )
		( 1.8399999999999999 , 1.0001914533888052 )
		( 1.85 , 1.0001553167644315 )
		( 1.8599999999999999 , 1.000124609886496 )
		( 1.87 , 1.000098241990801 )
		( 1.88 , 1.0000761451644333 )
		( 1.8900000000000001 , 1.0000578397015847 )
		( 1.9 , 1.0000428123128422 )
		( 1.9100000000000001 , 1.000030775655652 )
		( 1.92 , 1.0000213061067638 )
		( 1.9300000000000002 , 1.000014064904882 )
		( 1.94 , 1.0000087322173992 )
		( 1.9500000000000002 , 1.000004980504532 )
		( 1.96 , 1.0000025132520087 )
		( 1.97 , 1.0000010449649206 )
		( 1.98 , 1.000000256178351 )
		( 1.99 , 1.000000000011948 )
		( 1.9999999 , 1.0000000000000027 )
	};

			\addlegendentry[no markers, black]{BEP \cite{BEP11}}
			\addlegendentry[no markers, brown]{\cite{KS20, BF13}+ $\naive$ conversion}
			\addlegendentry[no markers, red]{\cite{KS20, BF13}+Theorem~\ref{thm:main_randomized}}

		\end{axis}
	\end{tikzpicture}
		\label{fig:vc_results}
	\end{subfigure}%
	\begin{subfigure}{.5\textwidth}
		\centering
		\caption{{\sc $3$-Hitting Set}}
		\input{figure_HS.tex}
		\label{fig:3hs_results}
	\end{subfigure}
	\caption{Results for {\sc Vertex Cover} and {\sc $3$-Hitting Set}.
	 A dot at $(\alpha,d)$ means that the respective algorithm outputs an $\alpha$-approximation in time $\OO^*(d^n)$.}
	\label{fig:runtimes}
\end{figure}

In this section we use Theorem~\ref{thm:main_randomized} to get faster randomized exponential approximation algorithms for {\sc Vertex Cover}, {\sc $3$-Hitting Set}, {\sc Directed Feedback Vertex Set ({\sc DFVS})}, {\sc Directed Subset Feedback Vertex Set ({\sc Subset DFVS})}, {\sc Directed Odd Cycle Transversal ({\sc DOCT})} and {\sc Undirected Multicut}.
These problems are defined in Appendix~\ref{sec:problem-definitions}.
One can observe that all these problems can be described as some $\phiminsubset$ problem.
Since all these problems can be interpreted as vertex deletion problems to some hereditary graph class, any parameterized $\alpha$-approximation algorithm for these problems can be used as an $(\alpha,\Phi)$-extension algorithm, for the respective $\Phi$.

\paragraph*{{\sc Vertex Cover ({\sc VC})}:}
In~\cite{BEP11} Bourgeois, Escoffier and Paschos designed several exponential approximation algorithms for {\sc VC} for approximation ratios in the range $(1,2)$.  
For any $\alpha\in (1,2)$ the best known running time of a parameterized randomized $\alpha$-approximation algorithm for {\sc VC} is attained in~\cite{KS20} if $\alpha\gtrsim 1.03$, and in~\cite{BF13} if $\alpha \lesssim 1.03$. 
We use these algorithms in conjunction with Theorem~\ref{thm:main_randomized} to obtain faster randomized  exponential $\alpha$-approximation algorithms for {\sc VC} for values of $\alpha$ in the range $(1,2)$.
We compare our running times to the running times obtained by the naive conversion (\hyperref[benchmark-2]{Benchmark 2}) and to the running times in~\cite{BEP11}.\footnote{The result of \cite{BEP11} provides an $\alpha$-approximation algorithm for every $\alpha\in(1,2)$. As the evaluation of these running times is not trivial, we only provide the running times which were explicitly given in \cite{BEP11} for selected approximation ratios.}
We present the running time for selected approximation ratios in Figure~\ref{fig:tables} and give a graphical comparison in Figure~\ref{fig:vc_results}.

\begin{figure}[h]
\begin{small}
\begin{center}
	{\sc Vertex Cover}
	\begin{tabular}{c||c|c|c|c|c|m{1.15cm}|m{1.2cm}|m{1.35cm}|m{1.6cm}|} 
		ratio & 1.1 &1.2 &1.3 &1.4 &1.5 &1.6 &1.7 &1.8 &1.9 
		 \\
		 \hline
		 $\brute(\alpha)$ &
		 1.716&1.583&1.496&1.433&1.385&1.347&1.317&
		 1.291&1.269
		 \\
		 \hline
		 BEP \cite{BEP11} &
		 1.127 & 1.099 & 1.083 & 1.069 & 1.056  & 1.043 & 1.032 & 1.021 & 1.01 
		 \\
		\hline
		\cite{KS20}+Naive Conv.&
		1.149&1.079&1.044&1.0236&1.0110&1.00469&1.00162&1.000406&1.0000432\\
		\hline 
		\cite{KS20}+Theorem~\ref{thm:main_randomized} &
		1.114&1.064&1.036&1.0203&1.0099&1.00435&1.00156&1.000397&1.0000428\\
		\hline
	\end{tabular}
\end{center}

\begin{center}
	{\sc $3$-Hitting Set}
	\begin{tabular}{ c||  c| c| c| c| c| c| c| c| c|} 
		ratio & 1.2 &1.4 &1.6 &1.8 &2.0 &2.2 &2.4 &2.6 &2.8 
		\\ \hline
		$\brute(\alpha)$ &
		1.583&1.433&1.347&1.291&1.251&1.220&1.196&1.177&1.162\\
		\hline
		\cite{KS20}+ Naive Conv.&
		1.471&1.196&1.105&1.0582&1.0326&1.0173&1.00831&1.00324&1.000903
		\\
		\hline 
		\cite{KS20}+Theorem~\ref{thm:main_randomized}&
		1.240&1.119&1.0698&1.0417&1.0248&1.0140&1.00711&1.00292&1.000853
		\\
		\hline
	\end{tabular}
\end{center}
\end{small}
\caption{Results for {\sc Vertex Cover} and {\sc $3$-Hitting Set}.
 A value $d$ at the column of an approximation ratio $\alpha$ means that the respective algorithm outputs an $\alpha$-approximation in time $\OO^*(d^n)$.}
\label{fig:tables}
\end{figure}

\paragraph*{{\sc $3$-Hitting Set ({\sc $3$-HS})}:} 
The problem admits a simple polynomial-time $3$-approximation algorithm which cannot be improved assuming UGC~\cite{KhotR08}.
For any $\alpha\in (1,3)$ the best known running time of a parameterized $\alpha$-approximation algorithm for {\sc $3$-HS} is attained by either~\cite{FKRS18} if $\alpha \lesssim 1.08$,  or~\cite{KS20} if $\alpha \gtrsim 1.08$.
Using these algorithms as parameterized extension algorithms,
we calculate the running times of $\alpha$-approximation algorithms for {\sc $3$-HS} attained using the naive conversion (\hyperref[benchmark-2]{Benchmark 2}) and Theorem~\ref{thm:main_randomized}, for values of $\alpha \in (1,3)$.
We provide the running times for selected approximation ratios in Figure~\ref{fig:tables} and a graphical comparison in Figure~\ref{fig:3hs_results}.

\paragraph*{{\sc DFVS}, {\sc Subset DFVS}, {\sc DOCT}, {\sc Undirected Multicut}:}
For all these problems~\cite{LokshtanovMRSZ21} gave parameterized $2$-approximation algorithms that run in time $\OO^*(c^k)$, for some constant $c >1$.
One can easily observe from the description of the {\sc DFVS} algorithm in~\cite{LokshtanovMRSZ21} that it runs in time $\OO^*(1024^k)$.
Using Theorem~\ref{thm:main_randomized} we get that {\sc DFVS} admits an exponential $2$-approximation algorithm that runs in time $\OO^*(1.2498^n)$.
This running time is significantly better than the running time derived using the naive conversion (\hyperref[benchmark-2]{Benchmark 2}) of the algorithm of~\cite{LokshtanovMRSZ21}, which does not give anything meaningful for this problem.
It is also significantly better than using \emls\ with the algorithm of~\cite{LokshtanovMRSZ21}, which gives $\OO^*((\emlsbound(1024))^n) = \OO^*(1.9991^n)$.
It is also qualitatively better than the brute-force $2$-approximation algorithm (\hyperref[benchmark-1]{Benchmark 1}), which runs in time $\OO^*(1.25^n)$.

Using Lemma~\ref{lem:comparison}, we can show that we get faster $2$-approximation algorithms for all mentioned problems compared to the brute-force $2$-approximation algorithm, or the naive conversion of the parameterized algorithms in~\cite{LokshtanovMRSZ21} or the application of \emls\ with the algorithms of~\cite{LokshtanovMRSZ21}.
Note that even though the algorithms derived from Theorem~\ref{thm:main_randomized} are only {\em qualitatively} better than brute-force approximation, we emphasize that \amls\ is always strictly better than brute-force approximation (and the other benchmarks described earlier).
Also, it reflects Part~\ref{item:comparison-1} of Lemma~\ref{lem:comparison}, that as $c$ increases (that is, as the parameterized extension algorithm becomes slower), our algorithm converges to the brute-force approximation.

\section{Analysis of Approximate Monotone Local Search}
\label{sec:analysis}
This section is dedicated to the proof of Theorem~\ref{thm:main_randomized}. 
As explained earlier, the algorithm promised in Theorem~\ref{thm:main_randomized} is $\amlsalgo$ (Algorithm~\ref{algo:final}).
In Lemma~\ref{lem:correctness} we prove the correctness of Algorithm~\ref{algo:final} and in Lemma~\ref{lem:runtime_f} we provide a formula for its running time.
Finally, in Lemma~\ref{lem:runtime_upper_bound} we upper bound the running time of the formula obtained in Lemma~\ref{lem:runtime_f} with $\OO^*(\amlsbound(\alpha,c)^{n})$, thereby proving Theorem~\ref{thm:main_randomized}.

\begin{lemma}[Correctness]\label{lem:correctness}
	$\amlsalgo$ (Algorithm~\ref{algo:final}) is a randomized $\alpha$-approximation algorithm for $\phiminsubset$.
\end{lemma}

\begin{proof}
Let $I\in \{0,1\}^{\star}$ be an instance and $\OPT$ be a minimum solution of $I$. 
Consider an execution of $\sample(I,k,t)$ (Algorithm~\ref{algo:intermediary}) in which $k=|\OPT|$ and $0\leq t \leq \alpha \cdot k$. 
If the algorithm selects a set $X$ in Step~\ref{int:select} such that $\abs{\OPT \cap X} \geq \frac{t}{\alpha}$, then $\abs{\OPT \setminus X} \leq k - \frac{t}{\alpha}$.
Furthermore, it holds that $(\OPT \setminus X) \cup X \in \F_I$, because $\OPT \cup X \supseteq \OPT \in \F_I$ and $\Phi$ is monotone by assumption.
Since $\Aext$ is an $(\alpha,\Phi)$-extension algorithm, given the input $(I, X, k - \left\lceil \frac{t}{\alpha} \right\rceil )$ it returns a set $Y$ such that $\abs{Y} \leq \alpha \cdot (k - \frac{t}{\alpha}) = \alpha k - t$ with constant probability $\gamma\in (0,1]$. Let $Z = X \cup Y \in \F_I$ as in Step~\ref{int:T} of Algorithm~\ref{algo:intermediary} and we get $\abs{Z} = \abs{X} + \abs{Y} \leq \alpha k $.
It follows that
\begin{equation}
\label{eq:sample_prob}
\begin{aligned}
	\Pr&\left(\sample(I, |\OPT|,t) \textnormal{ returns a set of size at most $\alpha \cdot |\OPT|$}\right)\\
	&\geq 
	\gamma\cdot  \Pr\left( \abs{X \cap \OPT} \geq  \frac{t}{\alpha}  \right) = \gamma \cdot \p(n,|\OPT|,t,t/\alpha),
\end{aligned}
\end{equation}
where $\p$ is the function defined in Section~\ref{sec:results}.

Now, consider the execution of Algorithm~\ref{algo:final} with $I$ as its input and let $S$ be the set returned by Algorithm~\ref{algo:final}.
It can be easily verified that $S\in \mathcal{F}_I$.
Also, if $\abs{\OPT} \geq \frac{n}{\alpha}$ then $|S|\leq n\leq \alpha \cdot \abs{\OPT}$ and therefore the algorithm returns an $\alpha$-approximate solution in this case.
We henceforth assume that $\abs{\OPT} < \frac{n}{\alpha}$.
Consider the iteration of the loop in Step~\ref{amls:loop} of Algorithm~\ref{algo:final} in which $k=|\OPT|$. 
By \eqref{eq:sample_prob}, at least one of the calls to Algorithm~\ref{algo:intermediary} in this iteration returns a set of size $\alpha \cdot |\OPT|$ or less  with probability at least
$$1 - \left(1 - \gamma\cdot  \p(n,k,t,t/\alpha)\right)^{1 / \p(n,k,t,t/\alpha)}\geq 1- \exp(-\gamma)>0,$$
where $t$ is the value selected in Step~\ref{amls:select_t}. 

Let $\sol$ be as in Algorithm~\ref{algo:final} at the end of the iteration in which $k=|\OPT|$.
The minimum cardinality set in $\sol$ has size at most $\alpha \cdot \abs{\OPT}$ with probability $1-\exp(-\gamma)$, thus the set $S$ returned by the algorithm satisfies $|S|\leq \alpha \cdot \abs{\OPT}$ with a constant probability.
\end{proof}

\begin{lemma}[Running time]\label{lem:runtime_f}
 $\amlsalgo$ (Algorithm~\ref{algo:final}) runs in time $f_{\alpha,c}(n) \cdot n^{\OO(1)}$ where
 \begin{equation}
  \label{eq:fdef}
  f_{\alpha,c}(n) \coloneqq \sum_{k = 0}^{\left\lfloor \frac{n}{\alpha} \right\rfloor} \min_{~t \in \left[0,\alpha k\right] \cap \mathbb{N}~} \frac{c^{k -  \frac{t}{\alpha} }}{\p\left(n,k,t,\frac{t}{\alpha} \right)}.
 \end{equation}
\end{lemma}

\begin{proof}
 For each choice of $k$, Algorithm~\ref{algo:final} chooses in Step \ref{amls:select_t} a number $t$ that minimizes $\frac{c^{k-\frac{t}{\alpha} }}{\p\left(n,k,t,\frac{t}{\alpha}\right)}$ which takes time $n^{\OO(1)}$.
 Then the algorithm calls Algorithm~\ref{algo:intermediary} $\p\left( n, k, t, \frac{t}{\alpha}\right)^{-1}$ times which takes $\frac{c^{k - \frac{t}{\alpha} } }{\p(n,k,t,\frac{t}{\alpha})} \cdot n^{\OO(1)}$ time in total.
 So overall, the running time is upper bounded by $f_{\alpha,c}(n) \cdot n^{\OO(1)}$. 
\end{proof}

We now proceed with the main part of the analysis which is to bound $f_{\alpha,c}(n)$ by $\amlsbound(\alpha,c)^{n}$ up to some polynomial factors.
We remark at this point (without giving a proof) that our analysis is in fact tight, that is, it can be shown that $f_{\alpha,c}(n)$ is equal to $\amlsbound(\alpha,c)^{n}$ up to some polynomial factors in $n$.

Recall that $\entropy(p)= -p \ln p -(1-p)\ln (1-p)$ denotes the entropy function.
We will use the following bound on binomial coefficients (see, e.g.,~\cite[Example~$11.1.3$]{TJ06}):
\begin{equation}\label{eq:binom}
	\frac{1}{n+1}\cdot  \exp\left( n\cdot \entropy\left(\frac{k}{n}\right)\right)\leq \binom{n}{k} \leq  \exp\left( n\cdot \entropy\left(\frac{k}{n}\right)\right)
\end{equation}
for all $n,k\in \mathbb{N}$ such that $0 \leq k \leq n$.

Moreover, we also need the following technical lemma.
Intuitively speaking, it says that small perturbations to the values of $a$ and $b$ do not change the value of $a \cdot \entropy\left( \frac{b}{a} \right)$ by a large amount. 

\begin{lemma}[$\star$]\label{lem:perturb}
	For $0 \leq b \leq a \leq n$ and $\varepsilon,\delta \in [-1,1]$ such that $a+\varepsilon \geq 0$ and $0\leq b+\delta \leq a+\varepsilon$, we have
	\begin{align*}
		\abs*{a\cdot \entropy\left( \frac{b}{a} \right) - (a + \varepsilon) \cdot \entropy\left( \frac{b + \delta}{a + \varepsilon} \right) } = \OO(\log(n)).
	\end{align*}
\end{lemma}

Finally, recall that $\D{a}{b}= a \ln \frac{a}{b} +(1-a)\ln \frac{1-a}{1-b}$ denotes the Kullback-Leibler divergence (see, e.g.,~\cite{TJ06}).

\begin{lemma}\label{lem:runtime_upper_bound}
	It holds that
	\begin{align*}
		f_{\alpha,c}(n) = \OO^*(\amlsbound(\alpha,c)^{n}).
	\end{align*}
\end{lemma}

\begin{proof}
Recall that $\p\left(n,k,t,\frac{t}{\alpha}\right)$ denotes the probability that a uniformly random set $X$ of $t$ elements out of $[n]$ satisfies that $|X \cap [k]| \geq \frac{t}{\alpha}$.
Thus, we have that
\begin{equation}
	\label{eq:hyper_bound}
	\p\left(n,k,t,\frac{t}{\alpha}\right) = \sum_{y\geq \left\lceil \frac{t}{\alpha} \right\rceil} \frac{ \binom{k}{y} \binom{n-k}{t-y}}{\binom{n}{t}} 
	\geq  
	 \frac{ \binom{k}{\left\lceil \frac{t}{\alpha} \right\rceil} \binom{n-k}{t-\left\lceil \frac{t}{\alpha} \right\rceil}}{\binom{n}{t}}  = 
	  \frac{ \binom{t}{\left\lceil \frac{t}{\alpha} \right\rceil} \binom{n-t}{k-\left\lceil \frac{t}{\alpha} \right\rceil}}{\binom{n}{k}},
\end{equation}
where the last equality holds since the distribution of $|X \cap [k]|$, where $X \subseteq [n]$ is a uniformly random set of cardinality $t$, is identical to the distribution of $|Y \cap [t]|$ where $Y\subseteq [n]$ is a uniformly random set of cardinality $k$.

Using \eqref{eq:hyper_bound} we have
\begin{equation}
	\label{eq:f_initial_bound}
\begin{aligned}
	f&_{\alpha, c}(n)=
 	\sum_{k = 0}^{\left\lfloor \frac{n}{\alpha} \right\rfloor} \min_{~t\in \left[0,\alpha k\right] \cap \mathbb{N}~} \frac{c^{k - \frac{t}{\alpha}}}{\p\left(n,k,t,\frac{t}{\alpha} \right)} \leq 
	\sum_{k = 0}^{\left\lfloor \frac{n}{\alpha} \right\rfloor} \min_{~t\in \left[0,\alpha k\right] \cap \mathbb{N}~} \frac{c^{k - \frac{t}{\alpha} } \cdot {n \choose k}}{{t \choose \left\lceil \frac{t}{\alpha} \right\rceil } {n - t \choose k - \left\lceil \frac{t}{\alpha} \right\rceil }} \\
	&\leq n^{\OO(1)}\cdot \sum_{k = 0}^{\left\lfloor \frac{n}{\alpha} \right\rfloor}{n \choose k} \exp \Bigg( \min_{~t\in \left[0,\alpha k\right] \cap \mathbb{N}~} \left( \left( k -\frac{t}{\alpha}  \right)\ln(c)
									   - t\cdot \entropy\left( \frac{\left\lceil \frac{t}{\alpha} \right\rceil}{t} \right)-(n-t)\cdot \entropy\left( \frac{k - \left\lceil \frac{t}{\alpha} \right\rceil}{n-t} \right) \right) \Bigg) \\
	&\leq  n^{\OO(1)}\cdot \sum_{k = 0}^{\left\lfloor \frac{n}{\alpha} \right\rfloor}{n \choose k} \exp \Bigg( \min_{~t\in \left[0,\alpha k\right] \cap \mathbb{N}~} \left( \left( k - \frac{t}{\alpha} \right)\ln(c)- t\cdot \entropy\left( \frac{1}{\alpha} \right)-(n-t)\cdot \entropy\left( \frac{k - \frac{t}{\alpha}}{n-t} \right) \right) \Bigg) 
\end{aligned}
\end{equation}
where the second inequality follows from \eqref{eq:binom} and the  third inequality follows from Lemma~\ref{lem:perturb}.   

Define
\begin{equation*}
	g_{n,k}(t) \coloneqq \left( k - \frac{t}{\alpha} \right)\ln(c) - t\cdot \entropy\left( \frac{1}{\alpha} \right)-(n-t)\cdot \entropy\left( \frac{k - \frac{t}{\alpha}}{n-t} \right).
\end{equation*}

By Lemma~\ref{lem:perturb} it holds that $|g_{n,k}(t) -g_{n,k}(t-\varepsilon)|= \OO(\log n)$ for any $t\in [0,  \alpha k]$ and $0\leq \eps \leq \min\{1,t\}$. 
Using this observation and~\eqref{eq:f_initial_bound} we get
\begin{equation}
	\label{eq:f_second_bound}
	f_{\alpha,c}(n) \leq n^{\OO(1)}\cdot \sum_{k = 0}^{\left\lfloor \frac{n}{\alpha} \right\rfloor}{n \choose k} \exp \left( \min_{~t\in \left[0,\alpha k\right]\cap \mathbb{N} } g_{n,k}(t) \right) \leq 
	n^{\OO(1)}\cdot \sum_{k = 0}^{\left\lfloor \frac{n}{\alpha} \right\rfloor}{n \choose k} \exp \left( \min_{~t\in \left[0,\alpha k\right] } g_{n,k}(t) \right).
\end{equation}
Observe that in the last term the range of $t$ is not restricted to integral values. 

Let $\delta^*=\frac{\amlsbound(\alpha,c) -1}{c-1}$.
By the definition of $\amlsbound$ it holds that $\delta^* \in (0,\frac{1}{\alpha})$ and $\D{\frac{1}{\alpha}}{\delta^*} = \frac{\ln(c)}{\alpha}$.
Define  $t^*(n,k) \coloneqq \frac{k - n\delta^*}{\frac{1}{\alpha} -  \delta^*}$, thus  $\frac{k-\frac{t^*(n,k)}{\alpha}}{n-t^*(n,k)} = \delta^*$ and $\D{\frac{1}{\alpha}}{ \frac{k-\frac{t^*(n,k)}{\alpha}}{n-t^*(n,k)}}=\frac{\ln c}{\alpha}$ for every $n\in \mathbb{N}$ and $k\in \mathbb{N}$.  
It can be verified that $g_{n,k}(t)$ is convex and has a global minimum at $t^*(n,k)$, though this observation is not directly used by our proof. 

For any $n\in \mathbb{N}$ and  $k\in \left[0,\frac{n}{\alpha}\right]\cap \mathbb{N}$ it holds that  $t^*(n,k) =  \frac{\alpha k(\frac{1}{\alpha} - \delta^*) + \alpha k  \delta^* - n\delta^*}{\frac{1}{\alpha} - \delta^*} \leq \alpha k$ since $\alpha k \leq n$.
Furthermore, $t^*(n,k)\geq 0$ if and only if $k\geq n\delta^*$.
Following this observation we partition the summation in \eqref{eq:f_second_bound} into two parts.  
Define 
$$A(n)=\sum_{k = 0}^{\left\lfloor n\cdot \delta^* \right\rfloor}{n \choose k} \exp \left( \min_{~t\in \left[0,\alpha k\right] } g_{n,k}(t) \right)  \textnormal{ and }B(n)= \sum_{k =\left\lfloor n\cdot \delta^* \right\rfloor + 1  }^{\left\lfloor\frac{n}{\alpha } \right\rfloor}{n \choose k} \exp \left( \min_{~t\in \left[0,\alpha k\right] } g_{n,k}(t) \right).$$ 
Thus, $f_{\alpha,c}(n)\leq n^{\OO(1)} \cdot \left(A(n)+B(n) \right)$. We bound each of the sums $A(n)$ and $B(n)$ separately. 

In order to bound $B(n)$ we use the following algebraic identity.
\begin{claim}
 \label{claim:g_second}
 It holds that
 \begin{equation*}
  g_{n,k}(t) = \left( k - \frac{t}{\alpha} \right)\ln(c) + t \cdot \D{\frac{1}{\alpha}}{\frac{k - \frac{t}{\alpha}}{n-t}}+k\cdot \ln\left(\frac{k - \frac{t}{\alpha}}{n-t} \right) +(n-k) \cdot \ln \left(1-\frac{k - \frac{t}{\alpha}}{n-t} \right)  .
 \end{equation*}
\end{claim}
\begin{claimproof}
 We have
 \begin{align*}
		g&_{n,k}(t) 
		= \left( k - \frac{t}{\alpha} \right)\ln(c) - t\cdot \entropy\left( \frac{1}{\alpha} \right)-(n-t)\cdot \entropy\left( \frac{k - \frac{t}{\alpha}}{n-t} \right) \\
		&=  \left( k - \frac{t}{\alpha} \right)\ln(c) - t\cdot  \entropy\left( \frac{1}{\alpha} \right) +\left( k-\frac{t}{\alpha}\right) \ln \left( \frac{k - \frac{t}{\alpha}}{n-t}\right) + \left( n-k-t\left(1-\frac{1}{\alpha} \right) \right)\ln \left( 1-\frac{k - \frac{t}{\alpha}}{n-t}\right)\\
		&=\left( k - \frac{t}{\alpha} \right)\ln(c)   +t\cdot \left( -\entropy\left(\frac{1}{\alpha}\right) - \frac{1}{\alpha } \cdot \ln \left(\frac{k - \frac{t}{\alpha}}{n-t} \right) -\left(1-\frac{1}{\alpha}\right)\cdot \ln \left( 1- \frac{k - \frac{t}{\alpha}}{n-t} \right)\right) \\
		&~~~~~~~~~~+k\cdot \ln\left(\frac{k - \frac{t}{\alpha}}{n-t} \right) +(n-k) \cdot \ln \left(1 - \frac{k - \frac{t}{\alpha}}{n-t} \right) 
		\\&=
		\left( k - \frac{t}{\alpha} \right)\ln(c) 
	 + t \cdot \D{\frac{1}{\alpha}}{\frac{k - \frac{t}{\alpha}}{n-t}}+k\cdot \ln\left(\frac{k - \frac{t}{\alpha}}{n-t} \right) +(n-k) \cdot \ln \left(1-\frac{k - \frac{t}{\alpha}}{n-t} \right)  .
 \end{align*}
 The second equality is a rearrangement of the terms. The last equality uses the identity
 \begin{equation*}
  \D{x}{y} = x\ln\left (\frac{x}{y}\right) + (1-x)\ln\left(\frac{1-x}{1-y}\right) = -\entropy(x) - x\ln(y) - (1-x)\ln(1-y).
 \end{equation*}
\end{claimproof}

For any $n\in \mathbb{N}$ and $k\in \left[n\delta^*, \frac{n}{\alpha}\right]\cap \mathbb{N}$ it holds that $0\leq t^*(n,k)\leq \alpha k$.
Thus,
\begin{align*}
	\min_{t \in [0, \alpha k]} g_{n,k}(t)& \leq g_{n,k}\left( t^*(n,k) \right) \\
	&= \left( k - \frac{t^*(n,k)}{\alpha} \right)\ln(c) + t^*(n,k) \cdot \D{\frac{1}{\alpha}}{\delta^*} +k\cdot \ln (\delta^*) + (n-k) \cdot \ln (1-\delta^*)\\
	&= k \cdot \ln(c) +k\cdot \ln (\delta^*) + (n-k) \cdot \ln (1-\delta^*) \\
	&= k \cdot \ln\left(\frac{c \cdot \delta^*}{1- \delta^*}\right)  + n \cdot \ln(1 - \delta^*),
\end{align*}
where the first equality uses Claim \ref{claim:g_second} and the second equality follows from $\D{\frac{1}{\alpha}}{\delta^*} = \frac{\ln(c)}{\alpha}$. 
Therefore,
\begin{equation}\label{eq:sum_ii}
	\begin{aligned}
		B(n)&=
		\sum_{k = \left\lceil n \delta^* \right\rceil +1}^{\left\lfloor \frac{n}{\alpha} \right\rfloor} {n \choose k}\cdot \exp\left( \min_{t \in [0, \alpha k]} g_{n,k}(t) \right) \leq  \sum_{k = 0}^{n} {n \choose k} \left( \frac{c \cdot \delta^*}{1 - \delta^*} \right)^{k} \left( 1 - \delta^* \right)^{n}  \\
		&=   (1- \delta^*)^{n} \left( \frac{c \cdot \delta^*}{1 - \delta^*} + 1 \right)^{n} =  \left( (c-1)\delta^* + 1 \right)^{n}
	\end{aligned}
\end{equation}
using the Binomial Theorem.

We now proceed to bound $A(n)$. 
For every $n\in \mathbb{N}$ and $0 \leq k \leq \delta^* n$ it holds that 
\begin{equation*}
	\min_{t \in [0, \alpha k]} g_{n,k}(t) \leq g_{n,k}(0) = k\cdot \ln(c) - n \cdot \entropy \left(\frac{k}{n}\right)\leq k\cdot \ln c -\ln \binom{n}{k},
\end{equation*}
where the last inequality follows from~\eqref{eq:binom}.
Therefore,
\begin{equation}\label{eq:sum_i}
	A(n)=\sum_{k = 0}^{\left\lfloor n \delta^* \right\rfloor} {n \choose k}\cdot \exp\left( \min_{t \in [0, \alpha k]} g_{n,k}(t) \right) \leq  \sum_{k = 0}^{\left\lfloor n \delta^* \right\rfloor} c^{k} \leq n\cdot (c^{\delta^*})^{n}.
\end{equation}

Finally, by using \eqref{eq:sum_ii} and \eqref{eq:sum_i}, we get
\begin{equation*}
	f_{\alpha, c }(n)\leq n^{\OO(1)} \cdot \left( A(n)+B(n)  \right) \leq n^{\OO(1)} \cdot \left( (c^{\delta^*})^{n} + \left( (c-1)\delta^* + 1 \right)^{n} \right) \leq n^{\OO(1)} \cdot \left((c-1)\delta^* + 1\right)^{n},
\end{equation*}
where the third inequality uses $c^{\delta^*} \leq (c-1)\delta^* + 1 $ which holds because $f(x) \coloneqq c^x - (c-1)x - 1$ is convex and has two roots at 0 and 1.
By the definition of $\delta^*$ it holds that $(c-1)\delta^* + 1 = \amlsbound(\alpha,c)$ and thus, $f_{\alpha,c} (n)\leq n^{\OO(1) }\cdot \amlsbound(\alpha,c)^n$.
\end{proof}

Finally, Theorem~\ref{thm:main_randomized} is implied by Lemmas~\ref{lem:correctness},~\ref{lem:runtime_f} and~\ref{lem:runtime_upper_bound}.

\section{Derandomization}
\label{sec:derand-short}
In this section, we show how to derandomize Algorithm \ref{algo:final}. In particular, we prove Theorem~\ref{thm:main_derandomized}.
As usual, let $(U_I,\F_I)$ be a set system and let $k,t,n \in \mathbb{N}$ be the variables from Algorithm~\ref{algo:final} and let $\alpha \geq 1$.
In order to derandomize the algorithm, it is sufficient to find a collection $\CC$ of subsets of $U_I$ of size $t$ such that, for every possible solution set $S \subseteq U_I$ of size $k$, there is some set $X \in \CC$ such that
\[|X \cap S| \geq \frac{t}{\alpha}.\]
We refer to such a family $\CC$ as an \emph{$(n,k,t,\frac{t}{\alpha})$-set-intersection-family} which is formally defined below.

\begin{definition}
 Let $U$ be a universe of size $n$ and let $p,q,r \geq 1$ such that $n \geq p \geq r$ and $n - p + r \geq q \geq r$.
 A family $\CC \subseteq \binom{U}{q}$ is a \emph{$(n,p,q,r)$-set-intersection-family} if for every $T \in \binom{U}{p}$ there is some $X \in \CC$ such that $|T \cap X| \geq r$.
\end{definition}

Given a \emph{$(n,k,t,\frac{t}{\alpha})$-set-intersection-family} $\CC$ we can derandomize Algorithm \ref{algo:final} by iterating over all choices $X \in \CC$ instead of repeatedly sampling a set $X$ uniformly at random.
Observe that the derandomized algorithm (for a fixed $k,t$) runs in time $\OO^*(|\CC| \cdot c^{k - \frac{t}{\alpha}})$.
Now, we define
\[\kappa(n,p,q,r) \coloneqq \frac{\binom{n}{q}}{\binom{p}{r} \cdot \binom{n - p}{q - r}}.\]

The following theorem computes the desired set-intersection-family of small size.

\begin{theorem}[$\star$]
\label{thm:family}
 There is an algorithm that, given a set $U$ of size $n$ and numbers $p,q,r \geq 1$ such that $n \geq p \geq r$ and $n - p + r \geq q \geq r$,
 computes an $(n,p,q,r)$-set-intersection-family of size $\kappa(n,p,q,r)\cdot2^{o(n)}$ in time $\kappa(n,p,q,r)\cdot2^{o(n)}$.
\end{theorem}

For the proof of Theorem~\ref{thm:family} we extend the arguments from \cite{FominGLS19} which provide such a result for the special case when $q = r$ (which corresponds to the case $\alpha = 1$).

\begin{proof}[Proof of Theorem~\ref{thm:main_derandomized}]
 We proceed analogously to the proof of Theorem~\ref{thm:main_randomized} with the following changes.
 In Algorithm \ref{algo:final}, Step \ref{amls:select_t} we define
 \[t \coloneqq \argmin_{t \in \left[ 0,\alpha k\right]\cap \mathbb{N}} \kappa\left(n,k,t,\left\lceil\frac{t}{\alpha}\right\rceil\right)c^{k-\frac{t}{\alpha}}\]
 and then compute an $(n,k,t,\frac{t}{\alpha})$-set-intersection-family $\CC$ of size $\kappa(n,p,q,r)\cdot2^{o(n)}$ using Theorem~\ref{thm:family}.
 Afterwards, we repeatedly execute Algorithm~\ref{algo:intermediary} for every $X \in \CC$ (instead of sampling $X$ uniformly at random).
 Repeating the analysis of Lemma \ref{lem:runtime_f} the running time is bounded by $\OO^*(f_{\alpha,c}(n))$ where,
 \[f_{\alpha,c}(n) = \sum_{k = 0}^{\left\lfloor \frac{n}{\alpha} \right\rfloor} \min_{~t\in \left[0,\alpha k\right] \cap \mathbb{N}~} \kappa\left(n,k,t,\left\lceil\frac{t}{\alpha}\right\rceil\right)c^{k-\frac{t}{\alpha}} \cdot 2^{o(n)}.\]
 As we already proved in Lemma \ref{lem:runtime_upper_bound} it holds that
 \begin{align*}
       \sum_{k = 0}^{\left\lfloor \frac{n}{\alpha} \right\rfloor} \min_{~t\in \left[0,\alpha k\right] \cap \mathbb{N}~} \kappa\left(n,k,t,\left\lceil\frac{t}{\alpha}\right\rceil\right)c^{k-\frac{t}{\alpha}}
  &= \sum_{k = 0}^{\left\lfloor \frac{n}{\alpha} \right\rfloor} \min_{~t\in \left[0,\alpha k\right] \cap \mathbb{N}~} \frac{c^{k - \left\lceil \frac{t}{\alpha} \right\rceil} \cdot {n \choose t}}{{k \choose \left\lceil \frac{t}{\alpha} \right\rceil } {n - k \choose t - \left\lceil \frac{t}{\alpha} \right\rceil }}\\
  &= \sum_{k = 0}^{\left\lfloor \frac{n}{\alpha} \right\rfloor} \min_{~t\in \left[0,\alpha k\right] \cap \mathbb{N}~} \frac{c^{k - \left\lceil \frac{t}{\alpha} \right\rceil} \cdot {n \choose k}}{{t \choose \left\lceil \frac{t}{\alpha} \right\rceil } {n - t \choose k - \left\lceil \frac{t}{\alpha} \right\rceil }}
  \leq \amlsbound(\alpha,c)^{n} \cdot n^{\OO(1)}
 \end{align*}
 which results in the running time stated in the theorem.
\end{proof}

\section{The Brute-Force Approximation Algorithm}
\label{sec:brute} 
In this section we show that one can design an $\alpha$-approximate variant of exhaustive search that runs in time $\OO^*\left(\left(\brute(\alpha)\right)^n\right)$, where $\brute(\alpha)= 1 + \exp(-\alpha \cdot \entropy\left(\frac{1}{\alpha}\right))$.
We complement this result by showing that $\OO^*\left(\left(\brute(\alpha)\right)^n\right)$ is the best possible running time of an $\alpha$-approximation algorithm for a subset minimization problem in the \emph{oracle model} defined below.

A (randomized) \emph{oracle $\alpha$-approximation minimum subset} algorithm takes as input a universe $U$ and receives a membership oracle to a \emph{monotone} family $\F \subseteq 2^U$.
The algorithm returns a set $S \in \F$ such that $|S| \leq \alpha \cdot \min\big\{|T|~\big|~T \in \F\big\}$ (with constant probability).

\begin{theorem}\label{thm:brute-force-lower}
 For every $\alpha \geq 1$, there is a randomized oracle $\alpha$-approximation minimum subset algorithm which runs in time $\OO^*((\brute(\alpha))^n)$.
 Additionally, there is a deterministic oracle $\alpha$-approximation minimum subset algorithm which runs in time $(\brute(\alpha))^{n + o(n)}$.\footnote{%
 A previous version of this paper claimed a deterministic oracle $\alpha$-approximation minimum subset algorithm which runs in time $\OO^*((\brute(\alpha))^n)$.
 However, the proof relies on existing results in an incorrect way, and we currently can only prove the slightly weaker statement given in the theorem.}

 Moreover, there is no randomized oracle $\alpha$-approximation minimum subset algorithm which uses $\OO^*(c^n)$ oracle queries for any $c<\brute(\alpha)$.
\end{theorem}

The proof of Theorem~\ref{thm:brute-force-lower} utilizes the technical bound proved in Lemma~\ref{lem:brute_bounds}.
The proof of Lemma~\ref{lem:brute_bounds} uses arguments that similar to the ones used in the proof of Lemma~\ref{lem:runtime_upper_bound}.

\begin{lemma}
 \label{lem:brute_bounds}
 For any $n\in \mathbb{N}$ and $\alpha \geq 1$ it holds that 
 \[n^{-\OO(1)} \cdot \left(\brute(\alpha)\right)^n \leq \max_{k\in \left[0, \frac{n}{\alpha}\right) \cap \mathbb{N}} \frac{ \binom{n}{k} }{\binom{\floor{\alpha k }}{k}} \leq n^{\OO(1)} \cdot \left(\brute(\alpha)\right)^n\]
\end{lemma}

\begin{proof}
	We first establish the upper bound. Observe that
	$$
	\begin{aligned}
	\max_{k\in \left[0, \frac{n}{\alpha}\right) \cap \mathbb{N}} \frac{ \binom{n}{k} }{\binom{\floor{\alpha k }}{k}} 
	&\leq	n^{\OO(1)}\cdot \max_{k\in \left[0, \frac{n}{\alpha}\right) \cap \mathbb{N}} \binom{n}{k} \cdot \exp \left(- \floor{\alpha \cdot k} \cdot \entropy \left(\frac{k} {\floor{\alpha k }}\right)\right)\\
	&\leq n^{\OO(1)}\cdot 	\max_{k\in \left[0, \frac{n}{\alpha}\right) \cap \mathbb{N}} \binom{n}{k} \cdot \exp \left(- \alpha \cdot k \cdot \entropy \left(\frac{1} {\alpha  }\right)\right)\\
	&\leq  n^{\OO(1)}\cdot 	\sum_{k=0}^{n} \binom{n}{k} \cdot \exp \left(- \alpha \cdot k \cdot \entropy \left(\frac{1} {\alpha  }\right)\right)\\
	&=n^{\OO(1)} \left( 1+\exp\left( -\alpha\cdot  \entropy\left(\frac{1}{\alpha}\right)\right)\right)^n \\
	&=n^{\OO(1)} \cdot \left(\brute(\alpha)\right)^n,
	\end{aligned} $$
where the first inequality is by \eqref{eq:binom}, and the second inequality is by Lemma~\ref{lem:perturb}.

The lower bound is established similarly, though it requires some additional sophistication. By \eqref{eq:binom} and Lemma~\ref{lem:perturb} we have
	$$
\begin{aligned}
	\max_{k\in \left[0, \frac{n}{\alpha}\right) \cap \mathbb{N}} \frac{ \binom{n}{k} }{\binom{\floor{\alpha k }}{k}} 
	&\geq n^{-\OO(1) }\cdot \max_{k\in \left[0, \frac{n}{\alpha}\right) \cap \mathbb{N}} \exp\left(n\cdot \entropy\left(\frac{k}{n}\right) -\floor{\alpha \cdot k } \cdot \entropy\left(\frac{k}{\floor{\alpha \cdot k} }\right) \right)\\
	&\geq  n^{-\OO(1) }\cdot 
	\max_{k\in \left[0, \frac{n}{\alpha}\right) \cap \mathbb{N}} \exp\left(n\cdot \entropy\left(\frac{k}{n}\right) -\alpha \cdot k  \cdot \entropy\left(\frac{1}{\alpha }\right) \right),
\end{aligned} $$
 Define $g(\lambda )= \entropy(\lambda) - \lambda \cdot \alpha \cdot \entropy\left( \frac{1}{\alpha}\right)$.
 As the  entropy function is concave it follows that $g$ is concave.
 Furthermore, $g(0)=g\left(\frac{1}{\alpha}\right) = 0$, thus $g$ has a global maximum in the interval $\left[ 0,\frac{1}{\alpha}\right]$.
 Therefore,
$$
\begin{aligned}
	\max_{k\in \left[0, \frac{n}{\alpha}\right) \cap \mathbb{N}} \frac{ \binom{n}{k} }{\binom{\floor{\alpha k }}{k}} 
	&\geq  n^{-\OO(1) }	\cdot \max_{k\in \left[0, \frac{n}{\alpha}\right) \cap \mathbb{N}} \exp\left(n\cdot g\left( \frac{k}{n}\right) \right) \\
	&\geq n^{-\OO(1) }	\cdot \max_{k\in \left[0, n\right] \cap \mathbb{N}} \exp\left(n\cdot g\left( \frac{k}{n}\right) \right) \\
	&= n^{-\OO(1) }\cdot 	\max_{k\in \left[0, n\right] \cap \mathbb{N}} \exp\left(n\cdot \entropy\left( \frac{k}{n}\right) -k\cdot \alpha\cdot  \entropy\left(\frac{1}{\alpha }\right) \right) \\
	&\geq n^{-\OO(1) }\cdot 	\max_{k\in \left[0, n\right] \cap \mathbb{N}}\binom{n}{k} \cdot \exp\left( -k\cdot \alpha\cdot  \entropy\left(\frac{1}{\alpha }\right) \right) \\
	&\geq n^{-\OO(1) }\cdot \frac{1}{n}\cdot 	\sum_{k=0}^{n}\binom{n}{k}\cdot  \exp\left( -k\cdot \alpha\cdot  \entropy\left(\frac{1}{\alpha }\right) \right)\\
	&=n^{-\OO(1)}\cdot \left(1+\exp\left(-\alpha \cdot \entropy\left( \frac{1}{\alpha }\right) \right)\right)^n \\
	&=n^{-\OO(1)}\cdot \brute(\alpha)^n.
\end{aligned} $$
 The third inequality is by \eqref{eq:binom} and the forth holds since the maximum is at least the average.  
\end{proof}

We now prove the three parts of Theorem~\ref{thm:brute-force-lower} separately.
To obtain the claimed algorithms the basic idea is to sample $\OO^*((\brute(\alpha))^n)$ random sets (of some size $k$) and show that the desired approximate solution is found with constant probability.
This algorithm can be derandomized by using Theorem~\ref{thm:family} for the special case when $p = r$.

\begin{lemma}
 For every $\alpha \geq 1$, there is a randomized oracle $\alpha$-approximation minimum subset algorithm which runs in time $\OO^*((\brute(\alpha))^n)$.
\end{lemma}

\begin{proof}
For every $k \leq n/\alpha$ the algorithm performs the following steps:
Let $t \coloneqq \floor{\alpha k}$ and $s \coloneqq \binom{n}{k}/\binom{t}{k}$.
We uniformly at random sample sets $X_{k,1},\dots,X_{k,s} \subseteq U$ of size $t$ and check whether they are contained in $\F$ using the oracle.
If no such set is found, the algorithm returns the entire universe.

We argue that this algorithm is a $\alpha$-approximation algorithm.
Let $\OPT \subseteq U$ be a solution set of minimum cardinality and let $k \coloneqq |\OPT|$.
If $k \geq n/\alpha$, then the algorithm is clearly correct since even returning the entire universe gives the desired approximation ratio.
So suppose that $k \leq n/\alpha$ and let $t \coloneqq \floor{\alpha k}$ and $s \coloneqq \binom{n}{k}/\binom{t}{k}$.
For every $i \in [s]$, the probability that $\OPT \subseteq X_{k,i}$ is exactly $1/s$.
Hence, with probability at least
\[1 - \left(1 - \frac{1}{s}\right)^s \geq 1 - \frac{1}{e} \geq \frac{1}{2},\]
there is some $i \in [s]$ such that $\OPT \subseteq X_{k,i}$.
Since $\F$ is monotone we get that $X_{k,i} \in \F$ and the algorithms returns a solution set of size at most $|X| \leq \alpha k$.

Finally, the algorithm clearly runs in time
\[\max_{k\in \left[0, \frac{n}{\alpha}\right) \cap \mathbb{N}} \frac{ \binom{n}{k} }{\binom{\floor{\alpha k }}{k}} \cdot n^{\OO(1)}.\]
From Lemma \ref{lem:brute_bounds} we get that the running time is bounded by $\left(\brute(\alpha)\right)^n \cdot n^{\OO(1)}$.
\end{proof}

As indicated above, the randomized algorithm from the previous lemma can be derandomized by using Theorem~\ref{thm:family} for the special case when $p = r$.
More concretely, we rely on the notion of covering families.
Let $k < t < n$ be natural numbers and recall $[n] \coloneqq \{1,2,\dots,n\}$.
An \emph{$(n,t,k)$-covering} is a family $\CC \subseteq \{X \mid X \subseteq [n], |X|=t\}$ such that, for every $S \subseteq [n]$ of size $|S| = k$, there is some $X \in \CC$ such that $S \subseteq X$.
To construct the deterministic algorithm for Theorem~\ref{thm:brute-force-lower}, we exploit the following result that constructs $(n,t,k)$-coverings of almost optimal size.

\begin{corollary}
	\label{cor:covering-family}
	There is an algorithm that, given $k < t < n$, computes an $(n,t,k)$-covering $\CC$ of size $\binom{n}{k}/\binom{t}{k} \cdot 2^{o(n)}$ in time $\binom{n}{k}/\binom{t}{k} \cdot 2^{o(n)}$.
\end{corollary}

Note that Corollary~\ref{cor:covering-family} is a direct consequence of Theorem~\ref{thm:family} by setting $q \coloneqq t$, $p \coloneqq k$ and $r \coloneqq k$.

\begin{lemma}
 For every $\alpha \geq 1$, there is a deterministic oracle $\alpha$-approximation minimum subset algorithm which runs in time $(\brute(\alpha))^{n + o(n)}$.
\end{lemma}

\begin{proof}
By renaming elements, we may assume without loss of generality that $U = [n]$.
For every $k \leq n/\alpha$ the algorithm computes an $(n,\floor{\alpha k },k)$-covering $\CC_k$ using Corollary \ref{cor:covering-family} and checks for every set $X \in \CC_k$ whether it is contained in $\F$ using the oracle.
Finally, the algorithm returns the smallest set it finds that is contained in $\F$.
If no such set is found, the algorithm returns the entire universe.

It is easy to see that this algorithm is an $\alpha$-approximation algorithm.
Indeed, let $\OPT \subseteq U$ be a solution set of minimum cardinality and let $k \coloneqq |\OPT|$.
If $k \geq n/\alpha$, then the algorithm is clearly correct since even returning the entire universe gives the desired approximation ratio.
So suppose that $k \leq n/\alpha$.
By definition of an $(n,\floor{\alpha k },k)$-covering there is some $X \in \CC_k$ such that $\OPT \subseteq X$ and $|X| = \floor{\alpha k} \leq \alpha k$.
Since $\F$ is monotone we get that $X \in \F$ and the algorithm returns a solution set of size at most $|X| \leq \alpha k$.

Finally, using Corollary~\ref{cor:covering-family}, observe that the algorithm runs in time
\[\max_{k\in \left[0, \frac{n}{\alpha}\right) \cap \mathbb{N}} \frac{ \binom{n}{k} }{\binom{\floor{\alpha k }}{k}} \cdot 2^{o(n)}.\]
From Lemma \ref{lem:brute_bounds} we get that the running time is bounded by $\left(\brute(\alpha)\right)^{n + o(n)}$.
\end{proof}

Finally, we turn to the lower bound in Theorem~\ref{thm:brute-force-lower}.

\begin{lemma}
 Let $\alpha \geq 1$.
 There is no randomized oracle $\alpha$-approximation minimum subset algorithm which uses $\OO^*(c^n)$ oracle queries for any $c<\brute(\alpha)$.
\end{lemma}

\begin{proof}
For every $n\in \mathbb{N}$, define
$\displaystyle \kappa(n) =  \argmax_{k\in \left[0, \frac{n}{\alpha}\right) \cap \mathbb{N}} \binom{n}{k} \cdot  \binom{\floor{\alpha k }}{k}^{-1}$.
Also, for $n\in \mathbb{N}$ define $\Fadv (n) = \{S\subseteq [n]~|~|S|>\alpha \kappa(n)\}$.
For every $n\in \mathbb{N}$ and $X\subseteq [n]$ define a family $\F(n,X) = \{S\subseteq [n]~|~X\subseteq S\} \cup \Fadv(n)$.
It can be easily observed that $\F(n,X)$  and $\Fadv(n)$ are monotone set families.
Our bound is based on the difficulty that algorithms have to distinguish between $\Fadv(n)$ and $\F(n,X)$.

Let $\A$ be a randomized oracle $\alpha$-approximation minimum subset algorithm.
Without loss of generality we assume $\A$ only returns a set $S$ if it  queried the oracle with that set (and got back a positive answer).
Let $q(n)$ be the maximal number of oracle queries the algorithm uses given a universe of size $n$.
As the algorithm is randomized, we use $R$ to denote the random sequence of  bits used by the algorithm.

Fix $n\in \mathbb{N}$.
For any $j\in [q(n)]$ the $j$-th query to the oracle is a function of the previous answers the algorithm received from the oracle and the sequence of random bits the algorithm uses.
Thus, there is a function
$S_j(R)$ which returns the $j$-th query the algorithm sends to the oracle, given that the algorithm gets an oracle to $\Fadv(n)$.  If the algorithm does not issue the $j$-th query given $R$ we arbitrarily define $S_j(R)=\emptyset$.

Let $X\subseteq [n]$ be a random set of size $\kappa(n)$ which is sampled uniformly (and independently of $R$).
Consider the execution of $\A$ with the universe $[n]$ and an oracle for $\F(n,X)$.
Define
\begin{equation}
	\label{eq:C_def}
	\CC(R)= \bigcup_{j=1}^{q(n)} \begin{cases}
	                              \left\{ T\subseteq [n]~| ~|T|=\kappa(n),~T\subseteq S_j(R) \right\} &\textnormal{if }|S_j(R)|\leq \alpha \cdot \kappa(n) \\
		                            \emptyset                                                           &\textnormal{otherwise}
	                             \end{cases}.
\end{equation}
If $X\notin \CC(R)$ then the answers the algorithm receives to its queries are identical to the answers it would have received if it was given an oracle $\Fadv(n)$.
It therefore  asks the same queries, and must return a set $S\in \Fadv(n)$.
As $\A$ is a randomized $\alpha$-approximation algorithm, there is $\gamma\in (0,1]$ such that $\A$ returns a set $S\in \F(n,X)$ which satisfies $|S|\leq \alpha \cdot |X|=\alpha \cdot \kappa(n)$ with probability at least $\gamma$.
As all the sets in $\Fadv$ have cardinality greater than $\alpha \cdot \kappa(n)$ it follows that $\Pr(X\notin \CC(R) ) \leq 1-\gamma$, or equivalently, $\Pr(X\in \CC(R) )\geq \gamma$.

By the definition of $\CC(R)$ \eqref{eq:C_def}, each query $S_j(R)$ adds at most $\binom{\floor{\alpha \cdot \kappa(n)}}{\kappa(n)}$ sets to $\CC(R)$.
Thus $\CC(R) \leq q(n) \cdot \binom{\floor{\alpha \cdot \kappa(n)}}{\kappa(n)}$.
Since $X$ is independent of $R$ we have,
$$
\gamma \leq \Pr(X\in \CC(R) )\leq \frac{  q(n) \cdot \binom{\floor{\alpha \cdot \kappa(n)}}{\kappa(n)}}{ \binom{n}{\kappa(n)}},
$$
and hence
$$q(n) \geq \gamma \cdot \frac{\binom{n}{\kappa(n)} }{  \binom{\floor{\alpha \cdot \kappa(n)}}{\kappa(n)}} = \gamma \cdot   \max_{k\in \left[0, \frac{n}{\alpha}\right) \cap \mathbb{N}} \frac{ \binom{n}{k} }{\binom{\floor{\alpha k }}{k}}
\geq  n^{-\OO(1)} \cdot \left(\brute(\alpha)\right)^n,$$
where the equality follows from the definition of $\kappa(n)$ and the last inequality follows from Lemma~\ref{lem:brute_bounds}. In particular, this implies that $\A$ does not use $\OO^*( c^n)$ oracle queries for any $c<\brute(\alpha)$.
\end{proof}

However, this introduces another sublinear term in the exponent of the running time.
Instead, in this special case, we can rely on existing results on covering families~\cite{Kuzjurin00}.

Let $k < t < n$ be natural numbers and recall $[n] \coloneqq \{1,2,\ldots ,n\}$.
An \emph{$(n,t,k)$-covering} is a family $\CC \subseteq \{X \mid X \subseteq [n], |X|=t\}$ such that, for every $S \subseteq [n]$ of size $|S| = k$, there is some $X \in \CC$ such that $S \subseteq X$.
To construct the algorithm for Theorem~\ref{thm:brute-force-lower}, we exploit known constructions of $(n,t,k)$-coverings of almost optimal size.

\section{Concluding Remarks}
We introduced and analyzed approximate monotone local search \amls\ which can be used to obtain faster exponential approximation algorithms from parameterized (extension) approximation algorithms for monotone subset minimization problems.
In particular, we obtain faster exponential approximation algorithms for {\sc Vertex Cover}, {\sc $3$-Hitting Set}, {\sc DFVS}, {\sc Subset DFVS}, {\sc DOCT} and {\sc Undirected Multicut} (for some approximation ratios).

The significance of \emls\ stems from the abundance of existing parameterized (extension) algorithms which can be used to obtain the state-of-art exponential algorithms for multiple problems.
\amls\ has a similar potential in the context of exponential approximation algorithms.
Thus, our result further emphasizes the importance of the already-growing field of parameterized approximation, by exhibiting its strong connections with exponential-time approximations.

Some interesting follow-up questions of our work are the following.

\begin{problem}
 Can an $\alpha$-approximate algorithm running in time $\OO^*((\brute(\alpha) -\eps)^n)$ be derived from a parameterized extension $\beta$-approximation algorithm for any $\beta > \alpha$?
\end{problem}

For example, for {\sc Directed Feedback Vertex Set} only a parameterized $2$-approximation algorithm running in time $\OO^*(c^k)$~\cite{LokshtanovMRSZ21} is currently available.
The question is whether this algorithm can also be used to obtain an exponential $1.1$-approximation algorithm that runs in time $\OO^*((\brute(1.1) -\eps)^n)$, for some $\eps >0$?

We also described and showed that the exhaustive search analog in the $\alpha$-approximate setting achieves the best possible running time of $\OO^*(\left(\brute(\alpha)\right)^n)$ when one only has access to a membership oracle for the problem.
Observe that for $\alpha=1$, \SETH\ asserts that $(\brute(1))^n=2^n$ is tight.

\begin{problem}
 Does there exist a monotone subset minimization problem for which there is no $\alpha$-approximation algorithm that runs in time $\OO^*\left(   \left(  \brute(\alpha) - \eps \right)^n \right)$, assuming \SETH?
\end{problem}

Recall that the \amls\ algorithm only uses random sampling and the given parameterized $\alpha$-approximation extension algorithm.
Another interesting lower bound question is the following.

\begin{problem}
 Can one show that the running time of \amls\ is tight (up to polynomial factors) when one is only given access to a membership oracle and a parameterized $\alpha$-approximation extension algorithm as a black-box?
\end{problem}

\bibliographystyle{plain}
\bibliography{refs}

\appendix

\newpage

\section{Properties of $\amlsbound(\alpha,c)$}
\label{app:comparison}
In this section we prove Lemma~\ref{lem:comparison}.
Recall that $\amlsbound(\alpha,c)$ is the unique value $\gamma \in \left(1, 1+\frac{c-1}{\alpha}\right)$ such that $\D{\frac{1}{\alpha}}{\frac{\gamma-1}{c-1}} =\frac{\ln c }{\alpha}$.
The proof of Lemma~\ref{lem:comparison} follows from Lemmas~\ref{lem:amls_vs_emls},~\ref{lem:mono_r},~\ref{lem:amls_converge}~\ref{lem:amls_vs_brute} and~\ref{lem:amls_vs_naive}.

First, Lemma~\ref{lem:amls_vs_emls} shows that $\amls$ is strictly faster than the trivial generalization of $\emls$ (\hyperref[benchmark-3]{Benchmark 3}) when $\alpha \neq 1$.
We remark that Lemma~\ref{lem:amls_vs_emls} can also be proved via Lemma~\ref{lem:mono_r} and the fact that $\amlsbound(1,c) = \emlsbound(\alpha)$.
However, we give a separate proof containing simpler arguments.

\begin{lemma}\label{lem:amls_vs_emls}
	For every $\alpha,c > 1$,	$\amlsbound(\alpha,c) < \emlsbound(c) < 2$.
	In particular, $\amlsbound(\alpha,c)$ is bounded as a function of $c$. 
\end{lemma}

\begin{proof}
 Recall that $\emlsbound(c) = 2-\frac{1}{c}$.
 Without loss of generality, we can assume that $\emlsbound(c) \leq 1 + \frac{c-1}{\alpha}$ which implies $c \geq \alpha$, otherwise the claim is trivially satisfied.
 Since $\D{\frac{1}{\alpha}}{r}$ is decreasing for $r \in (0,\frac{1}{\alpha})$, we have $\amlsbound(\alpha,c) < \emlsbound(c)$ if and only if

 \begin{equation*}
  \frac{\ln(c)}{\alpha} = \D{\frac{1}{\alpha}}{\frac{\amlsbound(\alpha,c) -1}{c-1}} > \D{\frac{1}{\alpha}}{\frac{\emlsbound(\alpha,c) -1}{c-1}} = \D{\frac{1}{\alpha}}{\frac{1 - \frac{1}{c}}{c-1}}.
 \end{equation*}

 Since $c \geq \alpha$, we have $\frac{\alpha-1}{\alpha}\cdot\frac{c}{c-1} \leq 1$.
 Thus,
 \begin{equation*}
  \D{\frac{1}{\alpha}}{\frac{1 - \frac{1}{c}}{c-1}} - \frac{\ln(c)}{\alpha} = \D{\frac{1}{\alpha}}{\frac{1}{c}} - \frac{\ln(c)}{\alpha} = \frac{1}{\alpha}\cdot\ln\left(\frac{1}{\alpha}\right) + \left( 1 - \frac{1}{\alpha} \right)\cdot \ln\left( \frac{\alpha-1}{\alpha}\cdot\frac{c}{c-1} \right) < 0.
 \end{equation*}
\end{proof}

Next we show another upper bound to $\amlsbound(\alpha,c)$ which will be used in the proofs of the upcoming lemmas.
\begin{lemma}\label{lem:amls_ub}
	For all $\alpha, c > 1$,
	\begin{equation*}
	\amlsbound(\alpha,c) < \frac{\alpha c}{1 + (\alpha-1)c}.
	\end{equation*}
\end{lemma}

\begin{proof}
 Let $\phi \coloneqq \frac{\alpha  c}{1+\left(\alpha -1\right) c} = \frac{1 + (\alpha - 1)c + (c-1)}{1 + (\alpha - 1)c} = 1 + \frac{c-1}{1 + (\alpha -1)c} < 1 + \frac{c-1}{\alpha}$ where the inequality holds because $c > 1$.
 Since $\D{\frac{1}{\alpha}}{r}$ is decreasing for $r \in (0,\frac{1}{\alpha})$, we get that $\amlsbound(\alpha, c) < \phi$ if and only if
 \[\frac{\ln(c)}{\alpha} = \D{\frac{1}{\alpha}}{\frac{\amlsbound(\alpha,c) - 1}{c-1}} > \D{\frac{1}{\alpha}}{\frac{\phi - 1}{c-1}}  = \D{\frac{1}{\alpha}}{\frac{1}{1 + (\alpha - 1)c}}.\]
 Also, we have
 \begin{align*}
	\D{\frac{1}{\alpha}}{\frac{1}{1 + (\alpha - 1)c}} - \frac{\ln(c)}{\alpha} &= \frac{1}{\alpha}\cdot\ln \! \left(\frac{1+\left(\alpha -1\right) c}{\alpha}\right) -\frac{\ln \! \left(c \right)}{\alpha}+\left(1-\frac{1}{\alpha}\right)\cdot \ln \! \left(\frac{1-\frac{1}{\alpha}}{1-\frac{1}{1+\left(\alpha -1\right) c}}\right) \\
										  &= \frac{1}{\alpha}\cdot\ln \! \left(\frac{1+\left(\alpha -1\right) c}{c \alpha}\right) +\left(1-\frac{1}{\alpha}\right)\cdot \ln \! \left(\frac{1-\frac{1}{\alpha}}{1-\frac{1}{1+\left(\alpha -1\right) c}}\right)\\
										  &=  \frac{1}{\alpha}\cdot \ln\left(\frac{1+\left(\alpha -1\right) c}{c \alpha}  \right) + \left( 1 - \frac{1}{\alpha}\right)\cdot \ln\left( \frac{1+\left(\alpha -1\right) c}{c \alpha} \right)   \\
										  &= \ln \! \left(\frac{1+\left(\alpha -1\right) c}{c \alpha}\right) =  \ln \! \left(\frac{\alpha c - c + 1}{c \alpha}\right) < 0
 \end{align*}
 where the last inequality holds because $c > 1$.
 In combination, this proves the lemma.
\end{proof}

The next lemma is about the partial derivatives of $\amlsbound(\alpha,c)$, which is required to show the monotonicity of $\amlsbound$ in Lemma~\ref{lem:mono_r}.

\begin{lemma}\label{lem:amls_diff}
	For every fixed $\alpha > 1$, $\amlsbound(\alpha,c)$ is differentiable as a function of $c$ for $c > 1$.
	Similarly, for every fixed $c > 1$, $\amlsbound(\alpha,c)$ is differentiable as a function of $\alpha$ for $\alpha > 1$.
\end{lemma}

\begin{proof}
	 Let us fix some $\alpha > 1$.
	 We set $\Omega_{\alpha} = \{(c,y) \mid 1 < c, 1 < y < 1 + \frac{c-1}{\alpha}\}$ and define $f \colon \Omega_{\alpha} \to \mathbb{R}$ such that
	 \begin{equation*}
		f\left( c,y \right) \coloneqq \D{\frac{1}{\alpha}}{\frac{y-1}{c-1}} - \frac{\ln(c)}{\alpha}.
	\end{equation*}

  Since $f$ is twice differentiable, in particular it is continuously differentiable, and hence its partial derivative with respect to $y$ is given by $\frac{\partial f(c,y)}{\partial y} = \frac{\alpha  \left(y -1\right)-c +1}{\left(c -y \right) \alpha  \left(y -1\right)}$ (generated automatically).

	Let $c_0 > 1$ and set $y_0 = \amlsbound(\alpha, c_0)$.
	Then $(c_0, y_0) \in \Omega_{\alpha}$ and therefore $\frac{\partial f(c_0, y_0)}{\partial y} \neq 0$.
	By the Implicit Function Theorem (see, e.g., \cite[Theorem 1.3.1]{KP02}),	there exists a neighborhood $U$ that contains $c_0$ and a continuously differentiable function $h \colon U \to \mathbb{R}$ such that $h(c_0) = y_0 = \amlsbound(\alpha, c_0)$ and $\D{\frac{1}{\alpha}}{\frac{h(c)-1}{c-1}} =\frac{\ln c }{\alpha}$, i.e., $h(c) = \amlsbound(\alpha,c)$ for all $c \in U$.
	In particular this means that $\amlsbound(\alpha,c)$ is differentiable at $c_0$.
	Since $c_0 > 1$ was arbitrary, this implies that $\amlsbound(\alpha,c)$ is differentiable on $(1,\infty)$ as a function of $c$.

	For the second part of the lemma, let us fix some $c > 1$.
	We set $\Phi_{c} = \{(\alpha, y) \mid 1 < \alpha, 1 < y < 1 + \frac{c-1}{\alpha}\}$ and define $g \colon \Phi_{c} \to \mathbb{R}$ such that
  \begin{equation*}
		g\left( \alpha,y \right) \coloneqq \D{\frac{1}{\alpha}}{\frac{y-1}{c-1}} - \frac{\ln(c)}{\alpha}
	\end{equation*}

  Again, since $g$ is twice differentiable, in particular it is continuously differentiable, and hence its partial derivative with respect to $y$ is given by $\frac{\partial g(\alpha,y)}{\partial y} = \frac{\partial f(c,y)}{\partial y} = \frac{\alpha  \left(y -1\right)-c +1}{\left(c -y \right) \alpha  \left(y -1\right)}$.

	Let $\alpha_0 > 1$ and set $y_0 = \amlsbound(\alpha_0, c)$.
	Then $(\alpha_0, y_0) \in \Phi_{c}$ and therefore $\frac{\partial g(\alpha_0, y_0)}{\partial y} \neq 0$.
	By the Implicit Function Theorem, there exists a neighborhood $U$ that contains $\alpha_0$ and a continuously differentiable function $h \colon U \to \mathbb{R}$ such that $h(\alpha_0) = y_0 = \amlsbound(\alpha_0, c)$ and $\D{\frac{1}{\alpha}}{\frac{h(\alpha)-1}{c-1}} =\frac{\ln c }{\alpha}$, i.e., $h(\alpha) = \amlsbound(\alpha,c)$ for all $\alpha \in U$.
	In particular this means that $\amlsbound(\alpha,c)$ is differentiable at $\alpha_0$.
	Since $\alpha_0 > 1$ was arbitrary, this implies that $\amlsbound(\alpha,c)$ is differentiable on $(1,\infty)$ as a function of $\alpha$.
\end{proof}

\begin{lemma}
	\label{lem:mono_r}
	For every fixed $\alpha > 1$, the function $\amlsbound(\alpha, c)$ is strictly increasing for $c \in (1,\infty)$.
	Also, for every fixed $c > 1$, the function $\amlsbound(\alpha,c)$ is strictly decreasing for $\alpha \in (1,\infty)$.
\end{lemma}

\begin{proof}
 Let us fix some $\alpha > 1$.
 By Lemma~\ref{lem:amls_diff}, $\amlsbound(\alpha,c)$ is differentiable as a function of $c$ for $c > 1$.
 By using the chain rule for differentiation on $\D{\frac{1}{\alpha}}{\frac{\amlsbound(\alpha,c)-1}{c-1}} - \frac{\ln c }{\alpha} = 0$ we get (by generating the derivative automatically), for $c >1$,
 \begin{equation}\label{eq:diff_amls}
  \frac{\partial \amlsbound(\alpha, c)}{\partial c} =  -\frac{\left(\left(1+\left(\alpha -1\right) c \right) \cdot \amlsbound(\alpha,c)-\alpha  c \right) \left(\amlsbound(\alpha,c)-1\right)}{c \left(c -1\right) \left(-\alpha  \cdot \amlsbound(\alpha,c )+c +\alpha -1\right)}.
 \end{equation}
 Since $\amlsbound(\alpha,c) < 1 + \frac{c-1}{\alpha}$ using Lemma ~\ref{lem:amls_ub} and $c > 1$, the denominator of \eqref{eq:diff_amls} is positive. Therefore $\frac{\partial \amlsbound(\alpha, c)}{\partial c} > 0$ if and only if $( (1+\left(\alpha -1\right) c ) \cdot \amlsbound(\alpha,c)-\alpha  c) < 0$ which is same as
 \begin{equation*}
  \amlsbound(\alpha,c) < \frac{\alpha  c}{1+\left(\alpha -1\right) c}
 \end{equation*}
 which holds by Lemma~\ref{lem:amls_ub}.
 This proves the first part of the lemma.

 For the second part consider a fixed $c > 1$.
 Again, by Lemma~\ref{lem:amls_diff} we know that $\amlsbound(\alpha,c)$ is differentiable as a function of $\alpha$ for $\alpha \in (1,\infty)$.
 Similarly, by using the chain rule for differentiation on $\D{\frac{1}{\alpha}}{\frac{\amlsbound(\alpha,c)-1}{c-1}} - \frac{\ln c }{\alpha} = 0$ we get (by generating the derivative automatically), for $\alpha > 1$,
 \begin{equation*}
  \frac{\partial \amlsbound(\alpha,c)}{\partial \alpha} = \frac{\left(\amlsbound(\alpha,c)-1\right) \left(-c +\amlsbound(\alpha,c)\right) \ln \! \left(\frac{\left(\alpha -1\right) \left(\amlsbound(\alpha,c)-1\right) c}{c -\amlsbound(\alpha,c)}\right)}{\alpha  \left(\alpha \cdot \amlsbound(\alpha,c)-\alpha -c +1\right)}
 \end{equation*}
 Since $\amlsbound(\alpha,c) < 1  + \frac{c-1}{\alpha} < c$ we have that $\frac{\partial \amlsbound(\alpha,c)}{\partial \alpha} < 0$ if and only if
 \begin{equation*}
	\ln \! \left(\frac{\left(\alpha -1\right) \left(\amlsbound(\alpha,c)-1\right) c}{c -\amlsbound(\alpha,c)}\right) < 0 \iff \amlsbound(\alpha,c) < \frac{\alpha  c}{\alpha  c -c +1}
 \end{equation*}
 which again holds by Lemma~\ref{lem:amls_ub}.
\end{proof}

The following lemma shows that the running time of $\amlsalgo$ converges to that of random sampling when $c$ goes to infinity. 

\begin{lemma}\label{lem:amls_converge}
	For every $\alpha > 1$, $\displaystyle{\lim_{c\rightarrow \infty} \amlsbound(\alpha, c) = \brute(\alpha)}$.
\end{lemma}

\begin{proof}
	Let us fix some $\alpha > 1$.  
	Since $\amlsbound(\alpha,c) \leq 2$ (Lemma~\ref{lem:amls_vs_emls}) and $\amlsbound(\alpha,c)$ monotonically increasing (Lemma~\ref{lem:mono_r}), there is some $L\in [1,2]$ such that $\lim_{c\rightarrow \infty} \amlsbound(\alpha, c)=L$.
	Using $\D{\frac{1}{\alpha}}{ \frac{\amlsbound(\alpha,c)-1}{c-1}} = \frac{\ln{c}}{\alpha}$  we have,
	$$ 
	\begin{aligned}
		0 &=  \lim_{c\rightarrow \infty }\left( \D{\frac{1}{\alpha}}{ \frac{\amlsbound(\alpha,c)-1}{c-1}} - \frac{\ln{c}}{\alpha} \right)\\
		&= \lim_{c\rightarrow \infty} 
		\left( 
		-\entropy\left(\frac{1}{\alpha} \right)  -\frac{1}{\alpha}\ln \left( \frac{\amlsbound(\alpha,c)-1}{c-1} \right) - \left( 1-\frac{1}{\alpha} \right) \ln \left(\frac{c-\amlsbound(\alpha,c)}{c-1} \right) -\frac{1}{\alpha} \ln c 
		\right) \\
		&=-\entropy\left(\frac{1}{\alpha} \right) +
		\lim_{c\rightarrow \infty} 
		\left( 
		-\frac{1}{\alpha}\ln \left( \amlsbound(\alpha, c)-1 \right) -
		\left( 1-\frac{1}{\alpha} \right) \ln \left(1+\frac{1-\amlsbound(\alpha,c)}{c-1}  \right)-\frac{1}{\alpha} \ln \left( 1+\frac{1}{c-1} \right)
		\right) \\
		&= -\entropy\left(\frac{1}{\alpha} \right) -
		\frac{1}{\alpha}\ln \left(L -1 \right).
	\end{aligned}
	$$
	Hence,  $\ln(L-1)=-\alpha\cdot  \entropy\left( \frac{1}{\alpha}\right)$ and 
	$$L=1+ \exp\left(-\alpha\cdot  \entropy\left( \frac{1}{\alpha}\right)\right) = \brute(\alpha).$$
	That is, $\lim_{c\rightarrow \infty} \amlsbound(\alpha, c) = L = \brute(\alpha)$.
\end{proof}

\begin{lemma}\label{lem:amls_vs_brute}
	For every $\alpha \geq 1$ and $c > 1$, $\amlsbound(\alpha,c) < \brute(\alpha)$.
\end{lemma}

\begin{proof}
	For $\alpha = 1$ we have $\amlsbound(1,c) = (2-\frac{1}{c}) < 2 = \brute(1)$, for all $c > 1$.
	Lemma~\ref{lem:amls_converge} and Lemma~\ref{lem:mono_r} together imply that $\amlsbound(\alpha,c) \leq \brute(\alpha)$ for all  $\alpha,c > 1$.
	Moreover, it holds that $\amlsbound(\alpha,c)<\amlsbound(\alpha,c+1)\leq \brute(\alpha)$.
\end{proof}

The lemma above means that $\amls$ always provides a running time which is strictly better than that of random sampling.
Now we will prove the same result for the naive conversion.

\begin{lemma}\label{lem:amls_vs_naive}
 For every $\alpha \geq 1$ and $c > 1$, $\amlsbound(\alpha,c) < \naive(\alpha,c)$.
\end{lemma}

\begin{proof}
	For $\alpha = 1$ we have that $\amlsbound(1,c) < \naive(1,c)$ because $\amlsbound(1,c)  - \naive(1,c) = 2 - c - \frac{1}{c} < 0$ where the last inequality follows from the fact that $x + \frac{1}{x} > 2$ for $x > 1$.

	The proof follows from the fact that $\amlsbound(\alpha, c) < \frac{\alpha c}{1 + (\alpha-1)c} = 1 + \frac{c-1}{1 + (\alpha - 1)c} \leq  c^{\frac{1}{\alpha}}$.
	The first inequality follows from Lemma~\ref{lem:amls_ub}.
	To prove the second inequality let us consider $c^{\frac{1}{\alpha}} - \frac{\alpha c}{1 + (\alpha-1)c}$.
	Since $1 + x \leq e^x$ for all $x \in \mathbb{R}$,  $c^{\frac{1}{\alpha}} = e^{\ln(c) \cdot \frac{1}{\alpha}} \geq 1 + \frac{\ln(c)}{\alpha} \geq 1 +  \frac{c-1}{\alpha c}$ where the last inequality holds because $\ln(x) \geq \frac{x-1}{x}$ for $x  \geq 1$. Since $c >1$, this implies that
	\begin{equation*}
		c^{\frac{1}{\alpha}} - \frac{\alpha c}{1 + (\alpha-1)c} \geq \frac{c-1}{\alpha c} - \frac{c-1}{1 + (\alpha- 1)c} \geq 0.
	\end{equation*}
\end{proof}

\section{Proof of Lemma~\ref{lem:perturb}}
\label{app:analysis}
We first prove the following technical lemma.

\begin{lemma}\label{lem:xlnx}
	Let $x \geq 0, \varepsilon \in [-1,1]$ such that $x + \varepsilon \geq 0$. Then
	\begin{align*}
		\abs{x\ln(x) - (x+\varepsilon)\ln(x+\varepsilon)} \leq 7 + \ln(x+1)
	\end{align*}
\end{lemma}

\begin{proof}
	We consider two cases depending on whether $x + \varepsilon$ is less than 2 or not.
	\begin{mycases}
		\item If $x + \varepsilon \leq 2$ then $x \leq 3$ and we have
			\begin{align*}
				\abs{x\ln(x) - (x+\varepsilon)\ln(x+\varepsilon)}  = \abs{x\ln(x)} + \abs{(x+\varepsilon)\ln(x + \varepsilon)} \leq 2\max_{y \in [0,3]} \abs{y\ln(y)} \leq 7
			\end{align*}
		\item In this case $x + \varepsilon > 2$ which implies $x > 1$.
		 Consider the function $f(y) = y\ln(y)$ which is convex.
		 Hence, we can use the first-order estimate to bound $f$, i.e., $f(y') \geq f(y) + f'(y) \cdot (y' - y)$ for all $y,y' > 0$.
		 Note that $f'(y) = 1 + \ln (y)$.
		 We use this fact to bound $\abs{x\ln(x) - (x+\varepsilon)\ln(x+\varepsilon)}$ as follows.
			 \begin{itemize}
				\item By setting $y \coloneqq x$ and $y' \coloneqq x+\varepsilon$ we get
				\begin{align*}
					x\ln(x) - (x + \varepsilon)\ln(x + \varepsilon) &\leq x\ln(x) - \Big( x\ln(x) + \varepsilon(1 + \ln(x)) \Big) \\
					                                                &= -\varepsilon (1 + \ln(x)) \leq (1 + \ln(x+1)).
				\end{align*}
				\item On the other hand, by setting $y \coloneqq x+\varepsilon$ and $y' \coloneqq x$ we get
					 \begin{align*}
						 x\ln(x) - (x + \varepsilon)\ln(x + \varepsilon) &\geq (x + \varepsilon)\ln(x + \varepsilon) - \varepsilon(1 + \ln(x + \varepsilon)) - (x + \varepsilon)\ln(x + \varepsilon) \\
                                                             &= - \varepsilon(1 + \ln(x + \varepsilon)) \geq -(1 + \ln(x+\varepsilon)) \geq -(1 + \ln(x+1)).
					\end{align*}
				
			\end{itemize}
	\end{mycases}
	By combining all the bounds above we obtain the lemma.
\end{proof}

\begin{proof}[Proof of Lemma~\ref{lem:perturb}]
 \begin{align*}
        &\abs*{a\cdot \entropy\left( \frac{b}{a} \right) - (a + \varepsilon) \cdot \entropy\left( \frac{b + \delta}{a + \varepsilon} \right) }\\
  =~    &\Big\lvert-b\ln(b) - (a-b)\ln(a-b) + a\ln(a) + (b+\delta)\ln(b + \delta)\\
        &+ (a + \varepsilon - b - \delta)\ln(a + \varepsilon - b - \delta) - (a + \varepsilon)\ln(a + \varepsilon)\Big\rvert\\
  \leq~ &\abs*{b\ln(b) - (b + \delta)\ln(b + \delta)} + \abs{(a-b)} + \abs{a\ln(a) - (a+\varepsilon)\ln(a + \varepsilon)}\\
        &+ \abs*{(a-b)\ln(a-b) - \left( a - b + (\varepsilon - \delta) \right)\ln\left( a - b + (\varepsilon - \delta) \right) }\\
  =~    &\OO(\log(n))
 \end{align*}
 where last step uses Lemma~\ref{lem:xlnx}.
\end{proof}

\section{Derandomization: Proof of Theorem~\ref{thm:family}}
In the section, we give a proof of Theorem \ref{thm:family}.
Towards this end, we first extend the definition of set-intersection-families to differentiate between a weak and a strong version.
Then, the basic strategy to prove Theorem~\ref{thm:family} is to first provide a simple, but slow algorithm that computes a strong set-intersection-family with the desired parameters.
Afterwards, we prove Theorem~\ref{thm:family} by executing the slow algorithm on subsets of the domain and combining the results in a suitable way.

\begin{definition}
 Let $U$ be a universe of size $n$ and let $p,q,r \geq 1$ such that $n \geq p \geq r$ and $n - p + r \geq q \geq r$.
 A family $\CC \subseteq \binom{U}{q}$ is a \emph{strong $(n,p,q,r)$-set-intersection-family} if for every $T \in \binom{U}{p}$ there is some $X \in \CC$ such that $|T \cap X| = r$.
 Also, a family $\CC \subseteq \binom{U}{q}$ is a \emph{(weak) $(n,p,q,r)$-set-intersection-family} if for every $T \in \binom{U}{p}$ there is some $X \in \CC$ such that $|T \cap X| \geq r$.
\end{definition}

Recall that we defined
\[\kappa(n,p,q,r) \coloneqq \frac{\binom{n}{q}}{\binom{p}{r} \cdot \binom{n - p}{q - r}}.\]

As indicated above, we start by proving the existence of a strong $(n,p,q,r)$-set-intersection-family of size $\kappa(n,p,q,r) \cdot n^{\OO(1)}$ and give a simple, but slow algorithm to compute such a family.

\begin{lemma}
 \label{la:set-intersection-family-existence}
 Let $U$ be a universe of size $n$ and let $p,q,r \geq 1$ such that $n \geq p \geq r$ and $n - p + r \geq q \geq r$.
 Then there is a strong $(n,p,q,r)$-set-intersection-family $\CC \subseteq \binom{U}{q}$ such that $|\CC| \leq \kappa(n,p,q,r) \cdot n^{\OO(1)}$.
 Moreover, such a family can be computed in time $\OO(8^n)$.
\end{lemma}

\begin{proof}
 Let $s \coloneqq \kappa(n,p,q,r) \cdot (p + 1) \cdot \ln(n)$.
 Pick $\CC = \{X_1,\dots,X_s\}$ where each set $X_i$ is chosen uniformly and independently at random from the set $\binom{U}{q}$.
 Let $T \in \binom{U}{p}$ be a target set and fix some $i \in [s]$.
 Then
 \[\Pr(|T \cap X_i| = r) = \frac{1}{\kappa(n,p,q,r)}.\]
 Since the sets $X_i$ are chosen uniformly and independently at random it follows that
 \[\Pr(\forall i \in [s]\colon |T \cap X_i| \neq r) = \left(1 - \frac{1}{\kappa(n,p,q,r)}\right)^s \leq e^{-(p + 1) \cdot \ln(n)} = \frac{1}{n^{p+1}}.\]
 Using the union bound, we conclude that
 \[
 \begin{aligned} \Pr&(\CC \text{ is not a strong $(n,p,q,r)$-set-intersection-family}) \\
 	&\leq \sum_{T\in \binom{U}{p} } \Pr\left( \forall i\in [s]:~ |T \cap X_i| \neq r\right) \\
 	&\leq \binom{n}{p} \cdot \frac{1}{n^{p+1}}
 	\\& \leq \frac{1}{n}.
 	\end{aligned}\]
 In particular, the probability that $\CC$ is a strong $(n,p,q,r)$-set-intersection-family is strictly positive and hence, there exists a strong $(n,p,q,r)$-set-intersection-family $\CC$ of size
 \[|\CC| \leq \kappa(n,p,q,r) \cdot (p + 1) \cdot \ln(n).\]
 
 To compute a strong $(n,p,q,r)$-set-intersection-family, we make use of a simple approximation algorithm for the \textsc{Set Cover} problem.
 In the \textsc{Set Cover} problem, the input consists of a universe $\mathcal{V}$ and a collection of subsets $\mathcal{S} \subseteq 2^{\mathcal{V}}$. 
 We say $C\subseteq \mathcal{S}$ is a Set Cover if $\bigcup_{S\in C} S = \mathcal{V}$, and the objective is to find a Set Cover of minimal size.
 It is well known that a simple greedy algorithm attains a $(1+\ln |\mathcal{V}|)$-approximation for \textsc{Set Cover} in time $\OO\left(|\mathcal{V} |\cdot \sum_{S\in \mathcal{S}} |S| \right)$ (see, e.g., \cite[Chapter 2]{Vazirani01}). 
 Let $\mathcal{V} \coloneqq \binom{U}{p}$ and
 \[\mathcal{S} \coloneqq \Bigg\{\Big\{T \in \binom{U}{p} ~\Big|~ |X \cap T| = r\Big\} ~\Bigg|~ X \in \binom{U}{q}\Bigg\}.\]
 The \textsc{Set Cover} instance $(\mathcal{V},\mathcal{S})$ has a solution of size $\kappa(n,p,q,r) \cdot (p + 1) \cdot \ln(n)$ by the argument above.
 Since $|\mathcal{V}|\leq 2^n$ and $\sum_{S\in \mathcal{S}} |S| \leq 2^n \cdot 2^n$,
 we can compute a \textsc{Set Cover} for $(\mathcal{V},\mathcal{S})$ of size
 \[t = \OO(\kappa(n,p,q,r) \cdot (p + 1) \cdot n \cdot \ln(n))\]
 in time $\OO(|\mathcal{V}|\cdot \sum_{S \in \mathcal S}|S|) = \OO(8^n)$ using the greedy algorithm.
\end{proof}

To speed up the computation of the set-intersection-family, the basic idea is apply Lemma \ref{la:set-intersection-family-existence} to domains of smaller size.
This can be achieved via families of pairwise independent functions.

Let $U$ be a universe of size $n$ and let $b$ be a positive integer.
A family $\CX$ of functions $f\colon U \rightarrow [b]$ is \emph{pairwise independent} if for every $i,j \in [b]$ and every $u \neq v \in U$ it holds that
\[\Pr_{f \in \CX}(f(u) = i \wedge f(v) = j) = \frac{1}{b^2},\]
where $f$ is chosen uniformly at random. 

\begin{theorem}[\cite{AlonBI86}]
 \label{thm:pairwise-independent-functions}
 There is a polynomial-time algorithm that, given $U$ and $b$, computes a pairwise independent family $\CX$ of functions $f\colon U \rightarrow [b]$ such that $|\CX| = \OO(|U|^2)$.
\end{theorem}

We also utilize the following auxiliary lemma.

\begin{lemma}
 \label{lem:binom}
 Let $n,n',k,k',d\in \mathbb{N}$ such that $|n-n'|\leq d$, $|k-k'| \leq d$, $k\leq n$ and $k'\leq n'$.
 Then 
 \[\binom{n}{k}\cdot (n+n')^{-3d}\leq \binom{n'}{k'}\leq \binom{n}{k}\cdot (n+n')^{3d}.\]
\end{lemma}

\begin{proof}
Let $A\in \binom{[n']}{k'}$. It holds that,
$$\Big| \big| A\cap [n] \big|- k'\Big| = \Big|\big|A\cap [n]\big| - |A| \Big| = \Big|A\cap \big([n']\setminus [n]\big)\Big|\leq \big|[n'] \setminus[n] \big|   \leq |n'-n| \leq d,$$
and thus 
$$\Big| \big|A\cap[n] \big| -k\Big| = \Big| \big|A\cap[n] \big|-k' +k'-k\Big|\leq \Big| \big|A\cap[n] \big|-k'\Big| +\big| k'-k\big|\leq 2d.$$
Therefore, there is some $C \in \binom{[n]}{k}$ such that $\Big| \big(A\cap [n]\big) \triangle C  \Big| \leq 2d$ (we use $\triangle$ to denote he symmetric difference between two sets).

It follows that $A$ can be written as
$A = (C \triangle D) \cup E$ where $D=  \big(A\cap [n]\big) \triangle C$ and $E= A\cap \left([n']\setminus [n]\right)$.
Thus every set $A\in \binom{[n']}{k'}$ can be written as $A= (C\triangle D)\cup E$ where $C\in \binom{[n]}{k}$, $D\subseteq [n]$ with $|D|\leq 2d$ and $E\subseteq [n']$ with $|E|\leq d$.
Hence,
\begin{equation}
	\label{eq:binom_first}
	\binom{n'}{k'} \leq \binom{n}{k}\cdot n^{2d}\cdot {\left(n'\right)}^{d}\leq \binom{n}{k}(n+n')^{3d}.
\end{equation}
As the above argument also holds if we replace the roles of $n$ and $k$ with that of $n'$ and $k'$, we also have
\begin{equation}
	\label{eq:binom_second}
	\binom{n}{k} \leq  \binom{n'}{k'}(n+n')^{3d}.
\end{equation}
The statement of the lemma immediately follows from \eqref{eq:binom_first} and \eqref{eq:binom_second}.
\end{proof}

Finally, reformulating Equation \eqref{eq:binom}, we use that
\begin{equation}
 \label{eq:binom-exp}
 \frac{1}{n+1}\left[\left(\frac{k}{n}\right)^{-\frac{k}{n}}\left(1- \frac{k}{n}\right)^{\frac{k}{n} - 1}\right]^n \leq \binom{n}{k} \leq \left[\left(\frac{k}{n}\right)^{-\frac{k}{n}}\left(1- \frac{k}{n}\right)^{\frac{k}{n} - 1}\right]^n
\end{equation}
for all $n,k\in \mathbb{N}$ such that $0 \leq k \leq n$.

Now, we are ready to prove Theorem \ref{thm:family} which we restate for the readers convenience.

\begin{theorem}[Theorem \ref{thm:family} restated]
 There is an algorithm that, given a set $U$ of size $n$ and numbers $p,q,r \geq 1$ such that $n \geq p \geq r$ and $n - p + r \geq q \geq r$,
 computes an $(n,p,q,r)$-set-intersection-family of size $\kappa(n,p,q,r)\cdot2^{o(n)}$ in time $\kappa(n,p,q,r)\cdot2^{o(n)}$.
\end{theorem}

\begin{proof}
 Suppose that $n$ is sufficiently large.
 Let $b \coloneqq \lceil \log n\rceil$ and $\lambda \coloneqq \frac{r}{p}$.
 The algorithm first constructs a pairwise independent family $\CX$ of functions $f\colon U \rightarrow [b]$ using Theorem \ref{thm:pairwise-independent-functions}.
 Observe that $|\CX| = \OO(n^2)$.
 Let $f \in \CX$ and $i \in [b]$.
 We define $U_{f,i} \coloneqq \{u \in U \mid f(u) = i\}$ and $n_{f,i} \coloneqq |U_{f,i}|$.
 We say that $f \in \CX$ is \emph{good} if
 \[\left|n_{f,i} - \frac{n}{b}\right| \leq \sqrt{n} \cdot b\]
 for all $i \in [b]$.
 For every good $f \in \CX$, every $i \in [b]$, and every $p',q',r'$ such that $n_{f,i} \geq p' \geq r'$ and $n_{f,i} - p' + r' \geq q' \geq r'$, we compute a strong $(n_{f,i},p',q',r')$-set-intersection-family $\CC(f,i,p',q',r')$ over universe $U_{f,i}$ using Lemma \ref{la:set-intersection-family-existence}.
 Observe that
 \[|\CC(f,i,p',q',r')| \leq \kappa(n_{f,i},p',q',r') \cdot n^{\OO(1)}.\]
 Also note that all families can be computed in time $n^{\OO(1)} \cdot 2^{\OO(n/\log n)} = 2^{o(n)}$.
 
 A tuple $\bar t = (f,(p_i)_{i \in [b]},(q_i)_{i \in [b]})$ is \emph{valid} if
 \begin{enumerate}[label = (\Roman*)]
  \item\label{item:valid-1} $f \in \CX$ is good,
  \item\label{item:valid-2} $\left|p_i - \frac{p}{b}\right| \leq \sqrt{n} \cdot b$ and $p_i \leq n_{f,i}$ for all $i \in [b]$,
  \item\label{item:valid-3} $\left|q_i - \frac{q}{b}\right| \leq 3\cdot \sqrt{n} \cdot b$ for all $i \in [b]$, and
  \item\label{item:valid-4} $n_{f,i} - p_i + r_i \geq q_i \geq r_i$ where $r_i \coloneqq \lceil\lambda p_i\rceil$.
 \end{enumerate}
 For every valid tuple $\bar t = (f,(p_i)_{i \in [b]},(q_i)_{i \in [b]})$ we define a set family $\CC(\bar t)$ which contains all sets of the form $Y = Y_1 \cup \dots \cup Y_b$
 where
 \[Y_i \in \CC(f,i,p_i,q_i,r_i)\]
 and $r_i \coloneqq \lceil\lambda p_i\rceil$ for all $i \in [b]$.
 Observe that $|Y| = \sum_{i \in [b]}q_i$.
 
 \begin{claim}
  Let $\bar t = (f,(p_i)_{i \in [b]},(q_i)_{i \in [b]})$ be a valid tuple.
  Then $|\CC(\bar t)| =\kappa(n,p,q,r)\cdot2^{o(n)}$.
 \end{claim}
 \begin{claimproof}
  We have that
  \begin{align*}
   |\CC(\bar t)| &\leq \prod_{i \in [b]}|\CC(f,i,p_i,q_i,r_i)|\\
                 &\leq \prod_{i \in [b]}\kappa(n_{f,i},p_i,q_i,r_i) \cdot n^{\OO(1)}\\
                 &=    \prod_{i \in [b]}\frac{\binom{n_{f,i}}{q_i}}{\binom{p_i}{r_i} \cdot \binom{n_{f,i} - p_i}{q_i - r_i}} \cdot n^{\OO(1)}\\
                 &=    n^{\OO(\log n)} \cdot \prod_{i \in [b]}\frac{\binom{n_{f,i}}{q_i}}{\binom{p_i}{r_i} \cdot \binom{n_{f,i} - p_i}{q_i - r_i}}.
  \end{align*}
  As $\bar{t}$ is valid it holds that
  $$\left| \sum_{i\in [b] } n_{f,i} -n\right| = \left| \sum_{i\in[b]} \left( n_{f,i} - \frac{n}{b}\right) \right|\leq \sum_{i\in[b]} \left| n_{f,i} - \frac{n}{b}\right| \leq \sqrt{n}\cdot b^2,$$
  and 
  $$ \left| \sum_{i\in [b] } q_i -q\right| = \left| \sum_{i\in[b]} \left( q_{i} - \frac{q}{b}\right) \right| \leq \sum_{i\in[b]} \left| q_{i} - \frac{q}{b}\right| \leq \sqrt{n}\cdot b^2.$$
  Using the above we have that
  \[\prod_{i \in [b]} \binom{n_{f,i}}{q_i} \leq \binom{\sum_{i \in b} n_{f,i}}{\sum_{i \in b} q_i} \leq \binom{n}{q} \cdot n^{\OO(\sqrt{n}\cdot b^2)} = \binom{n}{q} \cdot 2^{o(n)},\]
  where the first inequality is a simple combinatorial inequality and the second inequality is due to Lemma~\ref{lem:binom}.
  
  Also,
  \[r_i \leq \lambda p_i + 1 \leq \lambda\left(\frac{p}{b} + \sqrt{n} \cdot b\right) + 1 \leq \frac{\lambda p}{b} + 2 \cdot \sqrt{n} \cdot b\]
  and
  \[r_i \geq \lambda p_i \geq \lambda\left(\frac{p}{b} - \sqrt{n} \cdot b\right) \geq \frac{\lambda p}{b} - \sqrt{n} \cdot b .\]
  So, using Lemma~\ref{lem:binom} and Equation \eqref{eq:binom-exp}, we get that
  \begin{align*}
   \prod_{i \in [b]} \binom{p_i}{r_i} &\geq \prod_{i \in [b]} \left[\binom{\lceil p/b\rceil}{\lceil\lambda(p/b)\rceil} \cdot n^{-\OO(\sqrt{n} \cdot b)}\right]\\
                                      &\geq \prod_{i \in [b]} \left[\Big(\lambda^{-\lambda} \cdot (1 - \lambda)^{\lambda - 1}\Big)^{p/b}\right] \cdot 2^{-o(n)}\\
                                      &=    \Big(\lambda^{-\lambda} \cdot (1 - \lambda)^{\lambda - 1}\Big)^{p} \cdot 2^{-o(n)}\\
                                      &\geq \binom{p}{r} \cdot 2^{-o(n)}.
  \end{align*}
  Now let $\mu \coloneqq \frac{q-r}{n-p}$.
  Then, by similar arguments,
  \begin{align*}
   \prod_{i \in [b]} \binom{n_{f,i} - p_i}{q_i - r_i} &
   \geq \prod_{i\in [b]} \left[ \binom{ \lceil (n-p)/b\rceil }{ \lceil (q-\lambda \cdot p )/b\rceil}\right]\\
   &= \prod_{i \in [b]} \left[\binom{\lceil (n-p)/b\rceil}{\lceil\mu((n-p)/b)\rceil} \cdot n^{-\OO(\sqrt{n} \cdot b)}\right]\\
                                                      &\geq \prod_{i \in [b]} \left[\Big(\mu^{-\mu} \cdot (1 - \mu)^{\mu - 1}\Big)^{(n-p)/b}\right] \cdot 2^{-o(n)}\\
                                                      &=    \Big(\mu^{-\mu} \cdot (1 - \mu)^{\mu - 1}\Big)^{n-p} \cdot 2^{-o(n)}\\
                                                      &\geq \binom{n-p}{q-r} \cdot 2^{-o(n)}.
  \end{align*}
  Overall, combining all the inequalities proved above, we get that
  \[|\CC(\bar t)| \leq \frac{\binom{n}{q}}{\binom{p}{r} \cdot \binom{n - p}{q - r}} \cdot 2^{o(n)}\]
  as desired.
 \end{claimproof}
 
 Now, we let $\CC$ be the following set family.
 For every valid tuple $\bar t = (f,(p_i)_{i \in [b]},(q_i)_{i \in [b]})$ and every set $Y \in \CC(\bar t)$ we do the following.
 If $|Y| \leq q$ then we add an arbitrary extension $Y' \supseteq Y$ to $\CC$ such that $|Y'| = q$.
 Otherwise, $|Y| \geq q$ and we add every subset $Y' \subseteq Y$ to $\CC$ such that $|Y'| = q$.
 Since there at most $n^{\OO(\log n)}$ many valid tuples, the last claim implies that there are at most $\kappa(n,p,q,r)\cdot2^{o(n)}$ many choices for the set $Y$.
 Also,
 \[||Y| - q| \leq 3 \cdot \sqrt{n} \cdot b^2 = o\left(\frac{n}{\log n}\right).\]
 Hence, for each set $Y$, we add only $2^{o(n)}$ many sets to the family $\CC$.
 Overall, this ensures that $\CC$ has size at most $\kappa(n,p,q,r)\cdot2^{o(n)}$.
 Also, by the same arguments, the family $\CC$ can be computed in the same time.
 It remains to prove that $\CC$ is indeed a $(n,p,q,r)$-set-intersection-family.
 
 \begin{claim}
  \label{claim:good-function}
  Let $T \subseteq U$ such that $|T| = p$.
  Then there is some $f \in \CX$ such that
  \[\Big||U_{f,i}| - \frac{n}{b}\Big| \leq \sqrt{n} \cdot b \;\;\;\text{and}\;\;\; \Big||U_{f,i} \cap T| - \frac{p}{b}\Big| \leq \sqrt{n} \cdot b\]
  for every $i \in [b]$.
 \end{claim}
 \begin{claimproof}
  Let us fix some $i \in [b]$ and consider a random function $f$ uniformly sampled from $\CX$.
  For every $u \in U$ let $X_u$ denote the indicator variable that $f(u) = i$.
  Then $\Pr_{f \in \CX}(X_u = 1) = \frac{1}{b}$.
  So
  \[\E_{f \in \CX}(|U_{f,i}|) = \E_{f \in \CX}\left(\sum_{u \in U} X_u\right) = \sum_{u \in U} \E_{f \in \CX}(X_u) = \frac{n}{b}\]
  and
  \[\E_{f \in \CX}(|U_{f,i} \cap T|) = \E_{f \in \CX}\left(\sum_{u \in T} X_u\right) = \sum_{u \in T} \E_{f \in \CX}(X_u) = \frac{p}{b}.\]
  Since $\CX$ is pairwise independent the covariance between any pair of distinct indicator variables is $0$.
  Hence,
  \[\Var_{f \in \CX}(|U_{f,i}|) = \Var_{f \in \CX}\left(\sum_{u \in U} X_u\right) = \sum_{u \in U} \Var_{f \in \CX}(X_u) \leq n\]
  and
  \[\Var_{f \in \CX}(|U_{f,i} \cap T|) = \Var_{f \in \CX}\left(\sum_{u \in T} X_u\right) = \sum_{u \in T} \Var_{f \in \CX}(X_u) \leq n.\]
  By Chebyshev's inequality, we conclude that
  \[\Pr_{f \in \CX}\left(\Big||U_{f,i}| - \frac{n}{b}\Big| \geq \sqrt{n}\cdot b\right) \leq \frac{1}{b^2}\]
  and
  \[\Pr_{f \in \CX}\left(\Big||U_{f,i} \cap T| - \frac{p}{b}\Big| \geq \sqrt{n}\cdot b\right) \leq \frac{1}{b^2}.\]
  Hence, by the union bound,
  \[\Pr_{f \in \CX}\left(\exists i \in [b]\colon \Big||U_{f,i}| - \frac{n}{b}\Big| \geq \sqrt{n}\cdot b \text{ or } \Big||U_{f,i} \cap T| - \frac{p}{b}\Big| \geq \sqrt{n}\cdot b\right) \leq 2b \cdot \frac{1}{b^2} = \frac{2}{b} < 1.\]
  This means there exists some function $f \in \CX$ with the desired properties.
 \end{claimproof}
 
 Now, let us fix some $T \subseteq U$ such that $|T| = p$, and let $f \in \CX$ be the function from the last claim.
 Let $T_i \coloneqq U_{f,i} \cap T$ and $p_i \coloneqq |T_i|$ for all $i \in [b]$.
 Also let $r_i \coloneqq \lceil\lambda p_i\rceil$.
 
 \begin{claim}
  There are positive integers $q_1,\dots,q_b$ such that $\bar t = (f,(p_i)_{i \in [b]},(q_i)_{i \in [b]})$ is valid.
 \end{claim}
 \begin{claimproof}
  First observe that $f$ is good by Claim \ref{claim:good-function}.
  Also, $p_i \leq n_{f,i}$ by definition and $\left|p_i - \frac{p}{b}\right| \leq \sqrt{n} \cdot b$ by Claim \ref{claim:good-function} for all $i \in [b]$.
  So we need to argue that there numbers $q_1,\dots,q_b$ that satisfy Conditions \ref{item:valid-3} and \ref{item:valid-4}. Let $i\in [b]$.
  We have that
  \begin{equation}
   \label{eq:npr_bound}
   \begin{aligned}
    n_{f,i} - p_i + r_i &\geq n_{f,i} - p_i + \lambda p_i\\
                        &\geq \left(\frac{n}{b} - \sqrt{n} \cdot b\right) - \left(\frac{p}{b} + \sqrt{n} \cdot b\right) + \lambda\left(\frac{p}{b} - \sqrt{n} \cdot b\right)\\
                        &\geq \frac{n - p + \lambda p}{b} - 3\cdot \sqrt{n} \cdot b\\
                        &\geq \frac{n - p + r}{b} - 3\cdot \sqrt{n} \cdot b\\
                        &\geq \frac{q}{b} - 3\cdot \sqrt{n} \cdot b,
   \end{aligned}
  \end{equation}
  where the last inequality holds by the conditions of the Theorem~\ref{thm:family}. 
  As in Condition~\ref{item:valid-4}, we define $r_i = \lceil \lambda \cdot r_i\rceil$. It holds that
  \begin{equation}
   \label{eq:r_bound}
   \begin{aligned}
    r_i &\leq \lambda p_i + 1\\
        &\leq \lambda\left(\frac{p}{b} + \sqrt{n} \cdot b\right) + 1\\
        &\leq \frac{\lambda p}{b} + 2 \cdot \sqrt{n} \cdot b\\
        &=    \frac{r}{b} + 2 \cdot \sqrt{n} \cdot b\\
        &\leq \frac{q}{b} + 2 \cdot \sqrt{n} \cdot b.
    \end{aligned}
  \end{equation}
  
  Define $q_i = \max\left\{r_i,~\lceil \frac{q}{b}-3\sqrt{n}\cdot b \rceil \right\}$. By definition we have $q_i\geq r_i$. 
  Since  $ n_{f,i}\geq p_i$ it holds that $n_{f,i}-p_i+r_i \geq r_i$.
  Using the last inequality and  \eqref{eq:npr_bound} we have $q_i\leq n_{f,i} - p_i+r_i$. 
  By \eqref{eq:r_bound} it holds that $q_i \in \left(\frac{q}{b}-3\sqrt{n} \cdot b,~ \frac{q}{b}+3\sqrt{n}\cdot b \right)$. Overall, it holds that $q_i$ satisfies Conditions \ref{item:valid-3} and \ref{item:valid-4}.  
 \end{claimproof}
 
 Let $\bar t$ denote the valid tuple from the last claim.
 For every $i \in [b]$ there is a set $Y_i \in \CC(f,i,p_i,q_i,r_i)$ such that $|T_i \cap Y_i| = r_i$.
 Let $Y \coloneqq Y_1 \cup \dots \cup Y_b$.
 Then $Y \in \CC(\bar t)$.
 Also
 \[|Y \cap T| = \sum_{i \in [b]}|T_i \cap Y_i| = \sum_{i \in [b]}r_i \geq \sum_{i \in [b]}\lambda p_i = \lambda p = r.\]
 To complete the proof, observe that by construction there is some $Y' \in \CC$ 
 \[|Y' \cap T| \geq r\]
 since $r \leq q$.
\end{proof}

\section{Problem Definitions}
\label{sec:problem-definitions}

In this section, we give the problem definitions of all the problems discussed in the paper.

\defproblem{{\sc Vertex Cover ({\sc VC})}}{An undirected graph $G$.}{Find a minimum set $S$ of vertices of $G$ such that $G-S$ has no edges.}

\defproblem{{\sc $3$-Hitting Set ({\sc $3$-HS})}}{A universe $U$ and set family $\mathcal{F} \subseteq {U \choose \leq 3}$.}{Find a minimum set $S \subseteq U$ such that for each $F \in \mathcal{F}$, $S \cap F \neq \emptyset$.}

\defproblem{{\sc Directed Feedback Vertex Set ({\sc DFVS})}}{A directed graph $G$.}{Find a minimum set $S$ of vertices of $G$ such that $G-S$ is a directed acyclic graph.}

\defproblem{{\sc Directed Subset Feedback Vertex Set ({\sc Subset DFVS})}}{A directed graph $G$ and a set $T \subseteq V(G)$.}{Find a minimum set $S$ of vertices of $G$ such that $G-S$ has no directed cycle that contains at least one vertex of $T$.}

\defproblem{{\sc Directed Odd Cycle Transversal ({\sc DOCT})}}{A directed graph $G$.}{Find a minimum set $S$ of vertices of $G$ such that $G-S$ has no directed cycle of odd length.}

\defproblem{{\sc Undirected Multicut}}{An undirected graph $G$ and a set $\mathcal{P} \subseteq V(G) \times V(G)$.}{Find a minimum set $S$ of vertices of $G$ such that $G-S$ has no path from $u$ to $v$ for any $(u,v) \in \mathcal{P}$}

\end{document}